\documentclass[twocolumn,twoside]{IEEEtran}

\usepackage[hidelinks]{hyperref}

\usepackage{amsmath,amsfonts,amssymb}
\usepackage{xcolor}
\usepackage{graphicx}
\usepackage{tikz} 
\usetikzlibrary{decorations.pathreplacing,decorations.pathmorphing,decorations.text}
\usepackage{float}
\usepackage{amsthm}
\usepackage{blkarray}
\usepackage{multirow}

\usepackage{orcidlink}

\newtheorem{theorem}{Theorem}
\newtheorem{lemma}[theorem]{Lemma}
\newtheorem{proposition}[theorem]{Proposition}
\newtheorem{corollary}[theorem]{Corollary}
\newtheorem{definition}[theorem]{Definition}

\newtheorem{example}[theorem]{Example}

\newcommand{\tr}{\ensuremath\mathrm{tr}}

\newcommand{\bra}[1]{\ensuremath\langle{#1}|}%
\newcommand{\ket}[1]{{\ensuremath|{#1}\rangle}}%
\newcommand{\ketbra}[2]{\ket{#1}\bra{#2}}
\newcommand{\braket}[2]{{\langle{#1}|{#2}\rangle}}
\DeclareMathOperator{\Tr}{Tr}

\newcommand{\errorbasis}[2]{^{#1}_{#2}}
\newcommand{\circuittensor}[3]{C({#1})\errorbasis{#2}{#3}}
\newcommand{\circuitoperator}[1]{\mathbf{C}\left({#1}\right)}
\newcommand{\pauligroup}[0]{\mathcal{P}}

\begin{document}

\title{Quantum Circuit Tensors and Enumerators with Applications to Quantum Fault Tolerance}

\author{Alon~Kukliansky\orcidlink{0009-0003-6743-9018}\thanks{A. Kukliansky is at the U.S. Naval Postgraduate School.}
and Brad~Lackey\orcidlink{0000-0002-3823-8757}\thanks{B. Lackey is at Microsoft Quantum, Microsoft Corporation.}}

\maketitle
\begin{abstract}
    We extend the recently introduced notion of tensor enumerator to the circuit enumerator. We provide a mathematical framework that offers a novel method for analyzing circuits and error models without resorting to Monte Carlo techniques. We introduce an analogue of the Poisson summation formula for stabilizer codes, facilitating a method for the exact computation of the number of error paths within the syndrome extraction circuit of the code that does not require direct enumeration. We demonstrate the efficacy of our approach by explicitly providing the number of error paths in a distance five surface code under various error models, a task previously deemed infeasible via simulation. We also show our circuit enumerator is related to the process matrix of a channel through a type of MacWilliams identity.
\end{abstract}

\tableofcontents

\section{Introduction}

Quantum error correction is an essential part of scalable quantum computation. A key component of this is the fault tolerance of the circuits used to extract syndromes of quantum codes \cite{shor1996fault, gottesman1998theory, preskill1998fault, knill2004fault1}. Typically, the analysis of such circuits relies on direct simulation using Monte Carlo trials \cite{knill2004fault2, wang2009threshold}. Often lower bounds of the fault tolerant threshold are obtained by counting error paths \cite{dennis2002topological}. Otherwise, there are a few general methods for the analysis of fault tolerant protocols, see for example \cite{beverland2024fault} for one.

Our approach is based on tensor enumerators of \cite{cao2023quantum}. These provide a method to explicitly analyze circuits and error models without the use of Monte Carlo techniques. In particular, we have developed a formalism that can represent both unitary circuits and quantum channels as tensor objects, from which enumerators can be derived.

We provide an analogue of the Poisson summation formula for stabilizer codes that allows for rapid enumeration of error paths.
 
Generally, the computation of quantum weight enumerators and hence circuit enumerators, is exponential in the number of qubits. For small circuits this is feasible, however, for larger circuits, this is only tractable when the circuit is sparse or has a special structure, which is often the case for syndrome extraction circuits. In particular, we find we can explicitly count the error paths of a distance five surface code under a variety of error models, which would be impossible using simulation.

This paper is organized as follows: we provide the needed background and layout of our notation in \S{\ref{section:background}}. Next, in \S{\ref{section:ct_linear_map}} we formalize the circuit tensor and Choi state for a linear map, prove composition laws, and compute these for key examples. In \S{\ref{section:ct_of_q_channel}}, we define the Choi matrix and circuit tensor for any quantum channel, and prove the circuit tensor composition laws at this level of generality. \S{\ref{section:ct_examples}} contains many circuit tensor constructions for common quantum circuits, with a final example of the teleportation circuit. We incorporate noise sources into the quantum circuit in \S{\ref{section:noise-analysis}}, and show how to analyze them using circuit enumerators. We extend the teleportation example to include quantum and classical noise sources, calculate its circuit enumerator, and generate the error model of the output state. \S{\ref{section:Poisson-summation}} holds the Poisson summation theorem for circuit enumerators of stabilizer codes. Following, \S{\ref{section:app_ft}} shows how to use circuit enumerators to analyze a noisy syndrome extraction circuit for a stabilizer code. Finally, in \S{\ref{section:concliu}} we discuss our theory and its impact.

In an appendix, we show our circuit tensor is related to the process matrix of a channel through a type of MacWilliams identity.

\section{Background}\label{section:background}

Here we provide an accelerated presentation of the background material we need for this work. We discuss quantum and classical error bases on Hilbert spaces, the Pauli basis in particular, as a means for introducing circuit tensors and enumerators. We also present the Shor-Laflamme enumerators for stabilizer codes, as well as their tensor enumerators.

Much of this work holds in the context of a general error basis. However, as in \cite{cao2023quantum}, we restrict to ``nice error bases'' \cite{knill1996group, knill1996non} with Abelian index group \cite{klappenecker2002beyond}, defined as follows.

\begin{definition}[Hilbert space error basis]\label{definition:error-basis}
    An \emph{error basis} on a Hilbert space $\mathfrak{H}$ is a basis of unitary operators $\mathcal{E}$, which
    \begin{enumerate}
        \item contains the identity $I \in \mathcal{E}$;
        \item is trace-orthogonal $\Tr(E^\dagger F) = 0$ if $E\not= F$; and,
        \item satisfies $EF = \omega(E,F) FE$ for all $E,F \in \mathcal{E}$.
    \end{enumerate}
    Here, we take $\omega(E,F) \in \mathbb{C}$ to have $\omega(E,F)^r = 1$ for some fixed $r$.
\end{definition}

Let $\mathcal{E}$ be an error basis on a Hilbert space $\mathfrak{H}$ of dimension $q$. Fix any basis $\{\ket{x}\}_{x=1}^q$ of $\mathfrak{H}$. We define the \emph{Bell state} relative to $\mathcal{E}$ to be $\ket\beta = \frac{1}{\sqrt{q}}\sum_{x} \ket{x}\otimes \ket{x} \in \mathfrak{H} \otimes \mathfrak{H}$. We have the following result from \cite[Lemma VII.1]{cao2023quantum}.

\begin{lemma}\label{lemma:bell-state-stabilizer}
$\ketbra\beta\beta = \frac{1}{q^2}\sum_{E\in \mathcal{E}} E \otimes E^* \in L(\mathfrak{H}\otimes\mathfrak{H})$.
\end{lemma}

 Our goal is to analyze quantum circuits, particularly those used for quantum error correction, and so we focus on the Pauli basis $\mathcal{P}_q$ of $\mathbb{C}^q$, which for $q > 2$ is defined by
\begin{equation}
    \mathcal{P}_q = \{ X_q^\alpha Z_q^\beta \::\: \alpha,\beta = 0, \dots, q-1 \},
\end{equation}
where 
\begin{equation}\label{eq:quantum_error_basis}
    X_q = \sum_{x=0}^{q-1} \ketbra{x+1\:\text{(mod $q$)}}{x} \text{\ and\ } Z_q = \sum_{x=0}^{q-1} \zeta_q^x \ketbra{x}{x},
\end{equation}
with $\zeta_q = e^{2\pi i/q}$. For $q=2$ we instead take the usual Pauli operators $\mathcal{P}_2 = \{I,X,Y,Z\}$. When the local dimension $q$ is understood or is irrelevant, we drop it from our notation.

As is usual in quantum information, classical information is encoded as diagonal operators. In the case of states, these are density operators; in the case of our error bases, these will be powers of the Pauli $Z$-operator. For a classical system with data $\{0,\dots, N-1\}$, we have the Hilbert space $\mathbb{C}^N$ with error bases $\{Z_N^\alpha\}_{\alpha=0}^{N-1}$ with
\begin{equation}\label{eq:classical_error_basis_zn}
   Z_N = \begin{pmatrix}
   1 & 0 & \cdots & 0 \\
   0 & \zeta_N & \cdots & 0 \\
   \vdots & \vdots & \ddots & \vdots \\
   0 & 0 & \cdots & \zeta_N^{N-1}
   \end{pmatrix},
\end{equation}
where, as above $\zeta_N = e^{2\pi i/N}$, and hence $\alpha \in \{0, \dots, N-1\}$:
\begin{equation}\label{eq:powers_of_classical_error_basis_zn}
   Z_N^\alpha = \begin{pmatrix}
   1 & 0 & \cdots & 0 \\
   0 & \zeta_N^\alpha & \cdots & 0 \\
   \vdots & \vdots & \ddots & \vdots \\
   0 & 0 & \cdots & \zeta_N^{(N-1)\alpha}
   \end{pmatrix}.
\end{equation}

Systems that consist of multiple subsystems are obtained through the tensor product. For example, given a local system with Hilbert space $\mathfrak{H}$ and error basis $\mathcal{E}$, a quantum code of length $n$ will have as its Hilbert space $\mathfrak{H}^{\otimes n}$, for which we can take the error basis
\begin{equation}
    \mathcal{E}^n = \{ E_1\otimes \cdots \otimes E_n \::\: E\in\mathcal{E}\}.
\end{equation}

We will often be faced with the task of tracking phases introduced by commuting Pauli operators; to unify across local dimensions $q$ we introduce the notation, which is consistent with its use in Definition~\ref{definition:error-basis} for general error bases, 
\begin{equation}
    \omega(P,Q) = \tfrac{1}{q}\Tr(P^\dagger Q^\dagger P Q).
\end{equation}
In particular $\omega(P,Q) = 1$ if and only if $P$ and $Q$ commute.

\begin{definition}[Phase relative to the Pauli error basis]\label{def:rel_pauli_phase}
    For a $n$-qubit Pauli operator $P$, we define its phase relative to the Pauli error basis as $\mu$:
    \begin{equation}
    \mu(P) = \begin{cases}
     1 & P \in \mathcal{P}^n\\
    -1 & -P \in \mathcal{P}^n\\
    -i & iP \in \mathcal{P}^n\\
    i & -iP \in \mathcal{P}^n
    \end{cases}.
    \end{equation}
\end{definition}

Note that this also applies to hybrid quantum/classical systems. For example, when performing nondemolition measurement of a qubit state, the resulting system consist of the post-measurement quantum state and the classical bit obtain from the measurement readout. Hence the Hilbert space is $\mathbb{C}^2\otimes \mathbb{C}^2$ with error basis
\begin{equation}
    \{ P\otimes Z^\alpha \::\: P\in\{I,X,Y,Z\},\ \alpha\in\{0,1\}\}.
\end{equation}

All of our examples will be qubit circuits, $q=2$. In fact, these will mostly be Clifford circuits, built from the Hadamard $H$, phase gate $S$, and CNOT. But as mentioned in the introduction, the power of our methods will also allow us to treat the $\frac{\pi}{8}$-gate $T$, as well as non-unitary operations such as state preparation $\mathcal{SP}_{\ket{\psi}}$ and destructive and nondemolition measurement, $\mathcal{MD}_{P}$ and $\mathcal{MP}_{P}$, with respect to a Pauli operator $P$.

Measurement is an example of a general quantum channel. Our notation for a quantum channel is $\mathcal{A}:\mathfrak{H} \leadsto \mathfrak{K}$, where it is important to recognize that $\mathcal{A}$ is not a function on the Hilbert space itself, but rather on operators on the Hilbert space. Applied on a density operator $\rho$ on $\mathfrak{H}$, the channel has operator sum, or Choi-Kraus, form
\begin{equation}
    \mathcal{A}(\rho) = \sum_{j} A_j \rho A_j^\dagger,
\end{equation}
where the Kraus operators $\{A_j\}$ need not be unique \cite[\S{8.2}]{Nielsen_Chuang_2010}. The analysis of error models uses quantum channels; for example the decoherence channel is given by $\mathcal{D}(\rho) = \sum_{P\in\mathcal{P}} P\rho P^\dagger$, which is closes related to the uniform Pauli error channel
\begin{equation}
    \mathcal{D}_p(\rho) = (1-\tfrac{3p}{4}) \rho + \tfrac{p}{4}(X\rho X + Y\rho Y + Z\rho Z).
\end{equation}

Our main application is the fault-tolerance analysis of stabilizer codes. We will write $\mathfrak{C}$ for such a code, with its stabilizer denoted as $\mathcal{S} = \mathcal{S}(\mathfrak{C}) = \langle S_1, \cdots, S_{n-k}\rangle$ where the generators $S_j$ are assumed independent. As usual we call $n$ the length and $k$ the dimension of $\mathfrak{C}$. The syndrome $s = (s_1, \dots, s_{n-k}) \in \{\pm 1\}^{n-k}$, is the results of projective measurements of the stabilizer generators; for an error operator $P$ this means $s_j = \omega(P,S_j)$.

The normalizer of $\mathfrak{C}$ is defined as
\begin{equation}\label{eq:N-definition}
    \mathcal{N} = \mathcal{N}(\mathfrak{C}) = \{ P\in \mathcal{P}^n \::\: \omega(P,S) = 1 \text{ for all $S\in \mathcal{S}$}\}.
\end{equation}
Elements of the $\mathcal{N}$ preserve $\mathfrak{C}$, and so (when not in $\mathcal{S}$) are naturally associated with logical errors. The definition~\eqref{eq:N-definition} of $\mathcal{N}$ can also be viewed as stating that $\mathcal{N}$ is the dual to $\mathcal{S}$ where ``orthogonality'' in this context is commutativity. That is $\mathcal{N}$ is the set of all the Pauli operators that are orthogonal to $\mathcal{S}$ in this sense. While not obvious the converse holds,
\begin{equation}
    \mathcal{S} = \{ P\in \mathcal{P}^n \::\: \omega(P,N) = 1 \text{ for all $N\in \mathcal{N}$}\}.
\end{equation}
In particular, we will often make use of the relations
\begin{equation}\label{eq:stabilizer-duality1}
    \sum_{D\in\mathcal{S}(\mathfrak{C})} \omega(D,E) =\begin{cases}    
     2^{n-k} & \text{if $E \in \mathcal{N}(\mathfrak{C})$,}\\
     0 & \text{otherwise,}
     \end{cases}
\end{equation}
and
\begin{equation}\label{eq:stabilizer-duality2}
    \sum_{D\in\mathcal{N}(\mathfrak{C})} \omega(D,E) = \begin{cases}  
     2^{n+k} & \text{if $E \in \mathcal{S}(\mathfrak{C})$,}\\
     0 & \text{otherwise.}
     \end{cases}
\end{equation}

\begin{corollary}\label{crol:alpha_stab}
Let $\ket\psi \in (\mathbb{C}^2)^{\otimes n}$ be a stabilizer state with stabilizer group $\mathcal{S}$. Each element of $\mathcal{S}$ is a Pauli operator, but it need not be in the positive Pauli basis $\mathcal{P}^n$ as it might include a nontrivial phase $-1$, $i$, or $-i$. Yet, we always have $-I \not\in \mathcal{S}$ as $\ket\psi$ is a $+1$ eigenstate for all elements of $\mathcal{S}$. It follows then $\pm iP \not\in \mathcal{S}$ for any $P \in \mathcal{P}^n$ as otherwise $(\pm iP)^2 = -I \in \mathcal{S}$. Hence any element $S$ of $\mathcal{S}$ can only have phase $\pm 1$ relative to $\mathcal{P}^n$, meaning $\mu(S)\in\{1,-1\}$.
\end{corollary}

The quantum weight enumerators \cite{shor1997quantum}
\begin{align}
    \label{eq:Shor-Laflamme-A} A(z;M_1,M_2) &= \sum_{P\in \mathcal{P}^n} \Tr(PM_1)\Tr(PM_2) z^{\mathrm{wt}(P)},\\
    \label{eq:Shor-Laflamme-B} B(z;M_1,M_2) &= \sum_{P\in \mathcal{P}^n} \Tr(PM_1 PM_2) z^{\mathrm{wt}(P)}
\end{align}
capture information about errors; when $M_1 = M_2 = \Pi_{\mathfrak{C}}$, the orthogonal projection onto $\mathfrak{C}$, these enumerate the Pauli operators in $\mathcal{S}(\mathfrak{C})$ and $\mathcal{N}(\mathfrak{C})$ respectively of each weight.

These weight enumerators $A$ and $B$ above are connected by a quantum MacWilliams transform \cite{shor1997quantum}. After homogenizing $A(w,z) = w^n A(z/w)$, and similarly for $B$, this reads:
\begin{equation}
    B(w,z;M_1,M_2) = A\left(\tfrac{w+3z}{2}, \tfrac{w-z}{2}; M_1, M_2\right).
\end{equation}
This identity has been extended to local dimension $q>2$ \cite{rains1998quantum}, the quantum analogue of the complete enumerator \cite{hu2019complete, hu2020weight}, and vector and tensor enumerators over these \cite{cao2023quantum}.

In this work, our circuit tensor is related to the total tensor enumerators of \cite{cao2023quantum},
\begin{align}
    A(z;M_1,M_2) &= \sum_{P,Q\in \mathcal{P}^n} \Tr(PM_1)\Tr(QM_2) e_{P,Q},\\
    B(z;M_1,M_2) &= \sum_{P,Q\in \mathcal{P}^n} \Tr(PM_1 QM_2) e_{P,Q}.
\end{align}
Where $e_{P,Q}$ are formal basis elements of the vector space of  $4^n\times 4^n$ matrices indexed by a pair of Pauli operators, which is of dimension $4^{2n}$. Here we are interested in the operational character of this matrix, so we write these as $e\errorbasis{P}{Q}$.

As a final note, we will occasionally run into the trivial system. For example when preparing a quantum state, the initial system is trivial. The associated Hilbert space of this system is just $\mathbb{C}$, for which there is only a single error operator $1$. We will usually suppress the notation of this Hilbert space and error basis for readability. Namely, for the case of state preparation, we will simply write $e\errorbasis{}{P}$ instead of $e\errorbasis{1}{P}$.

\section{Circuit tensor of a linear map}\label{section:ct_linear_map}

To apply tensor enumerator methods to problems in fault-tolerance we will need to construct circuit tensors for general quantum channels. Nonetheless, we find it instructive to first define and analyze the analogous objects for just a single linear operator. In this section, we define the Choi state \cite{delfosse2023spacetime} and circuit tensor for a linear operator and prove composition laws for these. We also apply this formalism to some key examples such as state preparation and unitary operations. 

Let $\mathfrak{H}$ and $\mathfrak{K}$ be finite dimensional Hilbert spaces; we write $L(\mathfrak{H},\mathfrak{K})$ for the space of linear operators from $\mathfrak{H}$ to $\mathfrak{K}$.

\begin{definition}[Choi state of a linear operator]
    Let $A\in L(\mathfrak{H},\mathfrak{K})$ be any linear operator, $\mathcal{E}$ an error basis on $\mathfrak{H}$, and $\beta$ the Bell state of $\mathcal{E}$. Then the \emph{Choi state} of $A$ is 
    \begin{equation}
    \ket{T_A} = (I \otimes A)\ket\beta = \tfrac{1}{\sqrt{\dim(\mathfrak{H})}}\sum_{x} \ket{x} \otimes A\ket{x} \in \mathfrak{H} \otimes \mathfrak{K}.
    \end{equation}
\end{definition}

Note that the Choi state need not be properly normalized. In particular, $\braket{T_A}{T_B} = \frac{1}{\dim(\mathfrak{H})} \tr(A^\dagger B)$, and so orthonormality of Choi states is directly related to that of their underlying operators in the trace norm. In particular, a Choi state will be unit length if $A$ is an isometry.

\begin{example}[States and effects]\label{example:state-prep-choi-state}
    Let $\ket\psi \in \mathfrak{H}$ be a pure state. View this as the operation of state preparation: a linear map $\mathbb{C} \to \mathfrak{H}$ given by $\alpha \mapsto \alpha\ket\psi$. Then under the identification of $\mathbb{C}\otimes\mathfrak{H} = \mathfrak{H}$, one has
    $\ket{T_{\ket\psi}} = \ket\psi$.

    Dually, the effect $\bra\psi$ is a linear map $\mathfrak{H} \to \mathbb{C}$. Operationally, this is interpreted as post-selection upon measuring $\ket{\psi}$. Then 
    \begin{equation}
        \ket{T_{\bra\psi}} = \tfrac{1}{\sqrt{\mathrm{dim}(\mathfrak{H}})} \sum_{x} \ket{x}\braket{\psi}{x}.
    \end{equation}
\end{example}

The key link to tensor networks throughout what follows is that the composition of operators can be realized through the trace of tensors. This is of course already apparent in matrix-product states, see for example \cite{cirac2021matrix}. The following simple, yet central, result can be viewed as just expressing this from the perspective of the Jamio{\l}kowski-Choi duality \cite{zyczkowski2004duality,jiang2013channel}.

\begin{proposition}[Choi state for a compositaion of linear operator]\label{proposition:Choi:composition}
    Let $A \in L(\mathfrak{H},\mathfrak{K})$ and $B \in L(\mathfrak{K},\mathfrak{L})$, and $\ket\beta \in \mathfrak{K}\otimes\mathfrak{K}$ the Bell state. Then
    \begin{equation}
        (I_{\mathfrak{H}} \otimes \bra\beta \otimes I_{\mathfrak{L}})(\ket{T_A}\otimes \ket{T_B}) = \tfrac{1}{\dim(\mathfrak{K})}\ket{T_{BA}}.
    \end{equation}
\end{proposition}
\begin{proof}
    We compute, albeit with some abuse of notation,
\begin{align}
    \nonumber &(I \otimes \bra{\beta} \otimes I) (\ket{T_A}\otimes \ket{T_B})\\
    \nonumber &= \tfrac{\dim(\mathfrak{H})^{-0.5}}{\dim(\mathfrak{K})} \sum_x \bigg\langle xx\bigg|_{2,3} \left(\sum_{y,z} \ket{y} \otimes A\ket{y} \otimes \ket{z} \otimes B\ket{z}\right)\\
    \nonumber &= \tfrac{1}{\dim(\mathfrak{K})\sqrt{\dim(\mathfrak{H})}} \sum_{x,y,z} \ket{y} \otimes \bra{x} A \ket{y} \otimes \braket{x}{z} \otimes B\ket{z}\\
    \nonumber &= \tfrac{1}{\dim(\mathfrak{K})\sqrt{\dim(\mathfrak{H})}} \sum_{x,y} \ket{y} \otimes B\ket{x} \bra{x} A \ket{y}\\
    &= \tfrac{1}{\dim(\mathfrak{K})\sqrt{\dim(\mathfrak{H})}} \sum_{y} \ket{y} \otimes BA\ket{y} = \tfrac{1}{\dim(\mathfrak{K})} \ket{T_{BA}}.
\end{align}
\end{proof}

We introduce the \emph{circuit tensor} of a linear operator by using its Choi state. One can view this object as a matrix representation of the operator using elements of error bases for indices. It has appeared in the literature under various names, for example the ``channel representation'' (of a unitary) in \cite[\S{2.3}]{Gosset2014}.

\begin{definition}[Circuit tensor of a linear operator]
    Let $A \in L(\mathfrak{H},\mathfrak{K})$, and let $\mathcal{E}$ and $\mathcal{E}'$ be error bases of $\mathfrak{H}$ and $\mathfrak{K}$ respectively. The \emph{circuit tensor} of $A$ is
    \begin{equation}\label{eqn:circuit-tensor-definition}
    \circuittensor{A}{E}{E'} = \Tr((E^* \otimes E')\ketbra{T_A}{T_A}),
    \end{equation}
    where $E \in \mathcal{E}$ and $E'\in \mathcal{E}'$.
\end{definition}

As above if $\mathcal{E}$ and $\mathcal{E}'$ are error bases of $\mathfrak{H}$ and $\mathfrak{K}$ respectively, we write $e^E_{E'}$ for the matrix unit of $(E,E')$: the $\dim(\mathfrak{H})^2\times\dim(\mathfrak{K})^2$ matrix, indexed by elements of $\mathcal{E}$ and $\mathcal{E}'$, with $1$ in its $(E,E')$-entry and $0$ elsewhere. We can then write the circuit tensor in component-free form:
\begin{equation}
    \circuitoperator{A} = \sum_{E\in \mathcal{E},E'\in \mathcal{E}'} \Tr((E^* \otimes E')\ketbra{T_A}{T_A}) e^E_{E'}.
\end{equation}

\begin{example}[State and effects, continued]\label{example:state-prep-circuit-tensor}
In Example \ref{example:state-prep-choi-state} above, we saw the Choi state associated with preparing or post-selected on a state is essentially that state itself. The domain of state preparation is $\mathbb{C}$ whose only error basis is $\{1\}$. The circuit tensor of state preparation then has only a single row, and so we suppress the use of $1$ as its index. That is:
\begin{equation}\label{eq:state-prep-circuit-tensor}
    \circuittensor{\ket\psi}{}{E} = \Tr(E \ketbra{T_{\ket\psi}}{T_{\ket\psi}}) = \bra\psi E \ket\psi.
\end{equation}
In particular, state preparation of qubit state $\ket\psi$ in the Pauli error basis has the circuit tensor
\begin{equation}\label{eq:state_prep_ct}
\circuitoperator{\ket\psi} = e\errorbasis{}{I} + \bra\psi X \ket\psi e\errorbasis{}{X} + \bra\psi Y \ket\psi e\errorbasis{}{Y} + \bra\psi Z \ket\psi e\errorbasis{}{Z}.
\end{equation}
Preparation of Pauli eigenstates then have the circuit tensors
\begin{equation}\label{eq:ct:state_prep}
    \begin{array}{r@{\:}c@{\:}lr@{\:}c@{\:}l}
        \circuitoperator{\ket{+}} &=& e\errorbasis{}{I} + e\errorbasis{}{X}, &
        \circuitoperator{\ket{-}} &=& e\errorbasis{}{I} - e\errorbasis{}{X},\\
        \circuitoperator{\ket{+i}} &=& e\errorbasis{}{I} + e\errorbasis{}{Y}, &
        \circuitoperator{\ket{-i}} &=& e\errorbasis{}{I} - e\errorbasis{}{Y},\\
        \circuitoperator{\ket{0}} &=& e\errorbasis{}{I} + e\errorbasis{}{Z}, &
        \circuitoperator{\ket{1}} &=& e\errorbasis{}{I} - e\errorbasis{}{Z}.      
    \end{array}
\end{equation}

Dually, the range of post-selection on $\ket\psi$ is $\mathbb{C}$, and so its circuit tensor has only a single column:
\begin{align}
    \nonumber &\circuittensor{\bra\psi}{E}{} = \Tr(E^* \ketbra{T_{\bra\psi}}{T_{\bra\psi}})\\
    \nonumber &\quad = \tfrac{1}{\dim{\mathfrak{H}}}\sum_{xy} \braket{y}{\psi}\bra{y}E^*\ket{x}\braket{\psi}{x}\\
    &\quad = \tfrac{1}{\dim{\mathfrak{H}}}\sum_{xy} \braket{\psi}{x}\bra{x}E\ket{y} \braket{y}{\psi} = \tfrac{1}{\dim{\mathfrak{H}}} \bra\psi E \ket\psi.
\end{align}
And so for a qubit effect in the Pauli basis
\begin{equation}\begin{split}
&\circuitoperator{\bra\psi}\\ 
&\quad = \frac{1}{2}\left(e\errorbasis{I}{} + \bra\psi X \ket\psi e\errorbasis{X}{} + \bra\psi Y \ket\psi e\errorbasis{Y}{} + \bra\psi Z \ket\psi e\errorbasis{Z}{}\right).
\end{split} \end{equation}

\end{example}

\begin{example}[Circuit tensor of a Bell state]\label{ex:bell_sp}
Preparation of the Bell state $\ket \beta = \frac{1}{\sqrt{2}}(\ket{00} + \ket{11})$ has the circuit tensor
\begin{equation}\label{eq:ct_bell_sp}
     \circuitoperator{\ket\beta} = e\errorbasis{}{II} + e\errorbasis{}{XX} - e\errorbasis{}{YY} + e\errorbasis{}{ZZ}.
\end{equation}

\end{example} 

Both of the previous two examples illustrate the circuit tensors of stabilizer states and indicate a link between the tensors and the stabilizers associated with the state. In generality, let $\ket\psi \in (\mathbb{C}^2)^{\otimes n}$ be an $n$-qubit stabilizer state with stabilizer group $\mathcal{S}$. As in Corollary~\ref{crol:alpha_stab}, we have $\mu:\mathcal{S}\to \{+1,-1\}$ so that for each $S \in \mathcal{S}$ we have $\mu(S)S \in \mathcal{P}^n$. Thus by writing $\ketbra\psi\psi = \frac{1}{2^n}\sum_{S\in\mathcal{S}} S$, the circuit tensor for preparing $\ket\psi$ is
\begin{align}
    \circuittensor{\ket\psi}{}{E} &= \Tr(E\ketbra\psi\psi) = \frac{1}{2^n}\sum_{S\in\mathcal{S}} \Tr(ES)\nonumber\\
    &= \begin{cases}
     \mu(E) & \text{if $E \in \mathcal{S}$,}\\
     0 & \text{otherwise,}\end{cases}
\end{align}
as all nonidentity Pauli operators have trace zero. Therefore we have proven the following result:

\begin{proposition}[Circuit tensor of a stabilizer state]\label{proposition:stabilizer-state-tensor}
    Let $\ket\psi \in (\mathbb{C}^2)^{\otimes n}$ be a stabilizer state with stabilizer group $\mathcal{S}$, and $\mu:\mathcal{S}\to \{+1,-1\}$ be the phase relative to the Pauli basis, see Corollary~\ref{crol:alpha_stab}. Then 
    \begin{equation}
        \circuitoperator{\ket\psi} = \sum_{S\in\mathcal{S}} \mu(S) e\errorbasis{}{S}.
    \end{equation}
\end{proposition}

For a general linear operator $A \in L(\mathfrak{H},\mathfrak{K})$, we can expand its circuit tensor as follows.
\begin{align}\label{eq:circuit-tensor-to-B-enumerator}
    \nonumber &\circuittensor{A}{E}{E'} = \bra{T_A} E^*\otimes E' \ket{T_A}\\
    \nonumber &\quad = \tfrac{1}{\dim(\mathfrak{H})} \sum_{xy} (\bra{x} \otimes \bra{x}A^\dagger) (E^* \otimes E') (\ket{y} \otimes A\ket{y})\\
    \nonumber &\quad = \tfrac{1}{\dim(\mathfrak{H})} \sum_{xy} \bra{x}E^*\ket{y} \bra{x} A^\dagger E' A\ket{y}\\
    \nonumber &\quad = \tfrac{1}{\dim(\mathfrak{H})} \sum_{xy} \bra{y} E^\dagger \ket{x} \bra{x} A^\dagger E' A\ket{y}\\
    &\quad = \tfrac{1}{\dim(\mathfrak{H})} \Tr(E^\dagger A^\dagger E'A).
\end{align}
So, except for a different normalization, the circuit tensor has precisely the same form as the $\mathrm{B}$-tensor enumerator of \cite{cao2023quantum}. In that work however, the $\mathrm{B}$-tensor enumerator is only defined for Hermitian operators; here we see the circuit tensor is the natural generalization to any linear map $A$.

In Proposition~\ref{proposition:Choi:composition} above, The composition of operators had a simple, yet slightly awkward, representation in terms of Choi states. However, this becomes entirely natural in the language of circuit tensors. From \eqref{eq:circuit-tensor-to-B-enumerator} we see that
\begin{equation}\label{eq:ct_identity}
    \circuitoperator{I} = \sum_{E\in\mathcal{E}} e^E_E,
\end{equation}
 and hence the circuit tensor of the identity operator is the identity matrix. Moreover, the circuit tensor is multiplicative, albeit order-reversing, as shown by the following result.

\begin{theorem}[Composition of circuit tensors]
    For $A\in L(\mathfrak{H},\mathfrak{K})$ and $B\in L(\mathfrak{K},\mathfrak{L})$, we have $\circuitoperator{BA} = \circuitoperator{A}\circuitoperator{B}$. That is,
    \begin{equation}
        \circuittensor{BA}{E}{E'} = \sum_F \circuittensor{A}{E}{F}\: \circuittensor{B}{F}{E'}.
    \end{equation}
\end{theorem}
\begin{proof}
    We simply compute, using Proposition \ref{proposition:Choi:composition} above:
    \begin{align}
        \nonumber &\sum_F \circuittensor{A}{E}{F}\ \circuittensor{B}{F}{E'}\\
        \nonumber &= \sum_F \Tr((E^*\otimes F) \ket{T_A}\bra{T_A}) \Tr((F^*\otimes E') \ket{T_B}\bra{T_B})\\
        \nonumber &= \sum_F \Tr((E^*\otimes F \otimes F^* \otimes E') (\ket{T_A}\otimes \ket{T_B})(\bra{T_A}\otimes \bra{T_B}))\\
        \nonumber &= \dim(\mathfrak{K})^2 \Tr((E^*\otimes \ketbra{\beta}{\beta} \otimes E') \\
        \nonumber &~~~~~~~~~~~~~~~~~~(\ket{T_A}\otimes \ket{T_B})(\bra{T_A}\otimes \bra{T_B}))\\
        &= \Tr((E^* \otimes E') \ket{T_{BA}}\bra{T_{BA}}) = \circuittensor{BA}{E}{E'}, 
    \end{align}
    where we have also used $\ketbra{\beta}{\beta} = \frac{1}{\dim^2} \sum_F F\otimes F^*$.
\end{proof}

\begin{example}[Circuit tensor of Pauli operators]
    Any Pauli group has $PQ = \omega(P,Q) QP$. So Pauli operators have diagonal circuit tensors:
    \begin{align}
        \nonumber \circuitoperator{P} &= \frac{1}{\dim(\mathfrak{H})}\sum_{Q,Q'\in\mathcal{P}^n} \Tr(Q^\dagger P^\dagger Q' P) e\errorbasis{Q}{Q'}\\
        &= \sum_{Q\in\mathcal{P}^n} \omega(Q,P) e\errorbasis{Q}{Q}.
    \end{align}
    In particular, for the single qubit Pauli operators:
    \begin{align}
        \label{eq:circuit-tensor-X} \circuitoperator{X} &= e\errorbasis{I}{I} + e\errorbasis{X}{X} - e\errorbasis{Y}{Y} - e\errorbasis{Z}{Z}\\
        \label{eq:circuit-tensor-Y}\circuitoperator{Y} &= e\errorbasis{I}{I} - e\errorbasis{X}{X} + e\errorbasis{Y}{Y} - e\errorbasis{Z}{Z}\\
        \label{eq:circuit-tensor-Z}\circuitoperator{Z} &= e\errorbasis{I}{I} - e\errorbasis{X}{X} - e\errorbasis{Y}{Y} + e\errorbasis{Z}{Z}.
    \end{align}
\end{example}

Examining equations (\ref{eq:circuit-tensor-X},\ref{eq:circuit-tensor-Y},\ref{eq:circuit-tensor-Z}), we see that the circuit tensor behaves similarly to how each of these operators acts on the Bloch sphere. For example, on the Bloch sphere, the action of $X$ is a $\pi$-rotation about the $x$-axis, which is precisely \eqref{eq:circuit-tensor-X} except for the $e^I_I$ term.

Let us recall how the Bloch-sphere representation is constructed. Ordinarily, we would use the Pauli basis, but any error basis $\mathcal{E}$ of a Hilbert space $\mathfrak{H}$ suffices. Suppose we are given a unitary operator $U\in\mathcal{U}(\mathfrak{H})$. For each $E'\in \mathcal{E}$ we conjugate by $U$ to get an operator $U^\dagger E' U$. We then expand this operator as a linear combination of operators in our error basis
\begin{equation}\label{eq:Bloch-sphere-representation}
U^\dagger E' U = \sum_{E\in \mathcal{E}} c^{E}_{E'} E.
\end{equation}
This defines a matrix $c^{E}_{E'} = c^{E}_{E'}(U)$ which is our Bloch-sphere representation of $U$.

Now, from \eqref{eq:circuit-tensor-to-B-enumerator} above we have
\begin{align}
    \nonumber\circuittensor{U}{E}{E'} &= \tfrac{1}{\dim(\mathfrak{H})} \Tr(E^\dagger U^\dagger E' U)\\
    \label{eq:circuit-tensor-is-Bloch-sphere} &= \tfrac{1}{\dim(\mathfrak{H})} \sum_{F\in \mathcal{E}} c^{F}_{E'} \Tr(E^\dagger F) = c^{E}_{E'}.
\end{align}
Therefore the circuit tensor of a unitary operator is precisely its Bloch-sphere representation as defined above.

While $I\in \mathcal{E}$, we would not normally include this when building the Bloch-sphere representation of $U$. Indeed $I$ is exceptional in the construction, as can be seen in the following lemma.

\begin{lemma}
    For any unitary $U$ the coefficient $c\errorbasis{E}{I}$ in \eqref{eq:Bloch-sphere-representation} equals $1$ for $E=I$ and $0$ for any $E\in\{X,Y,Z\}$.
\end{lemma}
\begin{proof}
    When setting $E'=I$ in \eqref{eq:circuit-tensor-is-Bloch-sphere} we have:
    \begin{equation}\label{eq:lemma_proof_1_eq}\begin{split}
        c\errorbasis{E}{I} &= \tfrac{1}{\dim(\mathfrak{H})} \Tr(E^\dagger U^\dagger I U) \\
        &= \tfrac{1}{\dim(\mathfrak{H})} \Tr(E^\dagger) = \begin{cases}
            1 &  E=I,\\
            0 & E\in\{X,Y,Z\}.
        \end{cases}
    \end{split}\end{equation}
\end{proof}

Therefore the circuit tensor of a unitary operator decomposes as $\circuitoperator{U} = e\errorbasis{I}{I} + \sum_{E,E' \not= I} c\errorbasis{E}{E'} e\errorbasis{E}{E'}$, and it is only the second term we ordinarily call the Bloch-sphere representation of $U$.

\begin{example}[Circuit tensor of Clifford operators]\label{example:Clifford_ct}
    Clifford operators normalize the Pauli group, thus for any $n$ qubit Clifford $G$ and Pauli $P\in\mathcal{P}^n$ we know that $GPG^\dagger$ is a Pauli operator, and hence in $\mathcal{P}^n$ up to a phase. We can also write it as $\forall_G \forall_{P_1\in\mathcal{P}^n} \exists_{P_2\in\mathcal{P}^n} GP_1=\mu(GP_1G^\dagger)P_2G$, see Defenitation~\ref{def:rel_pauli_phase}.
    This relation leads the Clifford operators to have a signed permutation circuit tensor:
    \begin{equation}
        \circuitoperator{G} = \sum_{P \in \mathcal{P}^n} \mu(GPG^\dagger)e\errorbasis{P}{\pm GPG^\dagger}
    \end{equation}

    In particular, for the Hadamard gate we have~\cite[4.18]{Nielsen_Chuang_2010}:
    \begin{equation}
        \begin{array}{ccc}
           HZ=XH,  & HX=ZH,   & HY=-YH.
        \end{array}
    \end{equation}
    Hence it follows that
    \begin{equation}\label{eq:ct_H}
        \circuitoperator{H} = e\errorbasis{I}{I} + e\errorbasis{X}{Z} + e\errorbasis{Z}{X} -    e\errorbasis{Y}{Y}.
    \end{equation}
    The case is similar for the $S$ gate, and some combinations of it with the Hadammard gate. We leave the following computations to the reader:
    \begin{align}
            \circuitoperator{S} &= e\errorbasis{I}{I} + e\errorbasis{X}{Y} + e\errorbasis{Z}{Z} -   e\errorbasis{Y}{X},\\
            \circuitoperator{S^\dagger} &= e\errorbasis{I}{I} - e\errorbasis{X}{Y} + e\errorbasis{Z}{Z} +   e\errorbasis{Y}{X},\\
            \circuitoperator{S\circ H} &= 
            e\errorbasis{I}{I} + e\errorbasis{X}{Z} + e\errorbasis{Z}{Y} + e\errorbasis{Y}{X}, \label{eq:ct_S_H}\\
            \circuitoperator{H \circ S^\dagger} &= 
            e\errorbasis{I}{I} + e\errorbasis{X}{Y} + e\errorbasis{Z}{X} + e\errorbasis{Y}{Z}. \label{eq:ct_H_Sd}
        \end{align}
    \end{example}

\begin{example}[Circuit tensor of CNOT]

To calculate the circuit tensor for the CNOT gate, denoted $CX$, we will look at how each error in our error basis applied to any of the input legs can be commuted to the other side of the CNOT, similarly to what was done in Example~\ref{example:Clifford_ct}. We can use the CNOT commutation rules found in~\cite[Table 1]{gottesman1998heisenberg}, and expand them to include also $Y$. A circuit representation of the commutation rules can be seen in Fig.~\ref{fig:xyz_cnot_comm}, and their mathematical representation is: 
\begin{align}
  CX(I \otimes X ) &= (I \otimes X)CX \\
  CX(Z \otimes I ) &= (Z \otimes I)CX \\
  CX(I \otimes Z ) &= (Z \otimes Z)CX \\
  CX(Y \otimes I ) &= (Y \otimes X)CX \\
  CX(I \otimes Y ) &= (Z \otimes Y)CX
\end{align}

\begin{figure}[ht]
    \centering
    \includegraphics[width=\linewidth]{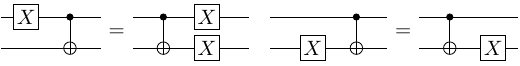}
    \includegraphics[width=\linewidth]{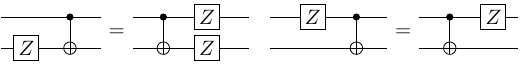}
    \includegraphics[width=\linewidth]{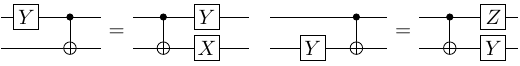}
    \caption[CNOT and Pauli commutation rules]{Commutation rules between Paulis and CNOT}
    \label{fig:xyz_cnot_comm}
\end{figure}

Using rules on Fig.~\ref{fig:xyz_cnot_comm} we can construct the full circuit tensor for the CNOT gate:
\begin{align}
\circuitoperator{CX} &=  e\errorbasis{I\otimes I}{I\otimes I}  + e\errorbasis{I\otimes X}{I\otimes X} + e\errorbasis{I\otimes Y}{Z\otimes Y} + e\errorbasis{I\otimes Z}{Z\otimes Z}  \nonumber\\
&+ e\errorbasis{X\otimes I}{X\otimes X} + e\errorbasis{X\otimes X}{X\otimes I} + e\errorbasis{X\otimes Y}{Y\otimes Z} - e\errorbasis{X\otimes Z}{Y\otimes Y}  \nonumber\\
&+ e\errorbasis{Y\otimes I}{Y\otimes X} + e\errorbasis{Y\otimes X}{Y\otimes I} - e\errorbasis{Y\otimes Y}{X\otimes Z} + e\errorbasis{Y\otimes Z}{X\otimes Y} \nonumber\\
&+ e\errorbasis{Z\otimes I}{Z\otimes I} + e\errorbasis{Z\otimes X}{Z\otimes X} + e\errorbasis{Z\otimes Y}{I\otimes Y} + e\errorbasis{Z\otimes Z}{I\otimes Z}
\label{eq:ct_cnot}
\end{align}
\end{example}

\begin{example}[Circuit tensor of $T$]
    It is straightforward to compute that $TXT^\dagger = \frac{1}{\sqrt{2}}(X + Y)$, $TYT^\dagger = \frac{1}{\sqrt{2}} (Y - X)$, and $TZT^\dagger = Z$. Therefore we can directly use \eqref{eq:Bloch-sphere-representation} and \eqref{eq:circuit-tensor-is-Bloch-sphere} to write out the circuit tensor for the $T$ gate:
    \begin{equation}\label{eq:ct:Tgate}
        \circuitoperator{T} = e\errorbasis{I}{I} + e\errorbasis{Z}{Z} + \frac{1}{\sqrt{2}}(e\errorbasis{X}{X} +e\errorbasis{X}{Y} - e\errorbasis{Y}{X} + e\errorbasis{Y}{Y})
    \end{equation}
\end{example}

In the above, we have focussed on the circuit tensor for quantum operations. Nonetheless, we can also compute circuit tensors for classical operations.

\begin{proposition}[Choi state and circuit tensor of a classical function]\label{proposition:clasical_ct}
For any classical function $f:\{0,\dots, N-1\} \to \{0, \dots, M-1\}$ define the operator
\begin{equation}
A_f = \sum_{x=0}^{N-1} \ket{f(x)}\bra{x} \in L(\mathbb{C}^N,\mathbb{C}^M).
\end{equation}

The Choi state of this operator, which for simplicity we denote $\ket{T_f}$, is
\begin{equation}    
\ket{T_f} = \frac{1}{\sqrt{N}}\sum_{x=0}^{N-1} \ket{x}\otimes A_f\ket{x} = \frac{1}{\sqrt{N}}\sum_{x=0}^{N-1} \ket{x}\otimes \ket{f(x)},
\end{equation}
and its circuit tensor, similarly with notation $\circuittensor{f}{}{}$ , is
\begin{equation}\begin{split}
    \circuittensor{f}{Z_N^\alpha}{Z_M^\beta} &= \Tr\left[(Z_N^{-\alpha}\otimes Z_M^\beta) \ketbra{T_f}{T_f}\right] \\&= \frac{1}{N} \sum_{x=0}^{N-1} \zeta_N^{-\alpha x} \zeta_M^{\beta f(x)},
\end{split}\end{equation}
where $\zeta_N=e^{\tfrac{2\pi i}{N}}$ is the canonical $N$-th root-of-unity.
\end{proposition}

\begin{proof}
\begin{equation}\begin{split}
    &\Tr\left[(Z_N^{-\alpha}\otimes Z_M^\beta) \ketbra{T_f}{T_f}\right] = \Tr\left[\bra{T_f}(Z_N^{-\alpha}\otimes Z_M^\beta) \ket{T_f}\right] 
    \\&\quad=\bra{T_f}(Z_N^{-\alpha}\otimes Z_M^\beta) \ket{T_f} 
    \\ &\quad= \frac{1}{N}\sum_{x=0}^{N-1}\left[\bra{x}\otimes \bra{f(x)}\right]\left[Z_N^{-\alpha}\otimes Z_M^\beta\right]\left[\sum_{y=0}^{N-1}\ket{y}\otimes \ket{f(y)}\right]
    \\ &\quad= \frac{1}{N}\sum_{x=0}^{N-1} \sum_{y=0}^{M-1}\left[\bra{x} Z_N^{-\alpha} \ket{y}\right] \otimes \left[\bra{f(x)} Z_M^{\beta} \ket{f(y)}\right]
    \\ &\quad= \frac{1}{N}\sum_{x=0}^{N-1} \sum_{y=0}^{M-1}\left[\bra{x} Z_N^{-\alpha} \ket{y}\right] \cdot \left[\bra{f(x)} Z_M^{\beta} \ket{f(y)}\right]
    \\ &\quad= \frac{1}{N}\sum_{x=0}^{N-1} \left[\bra{x} Z_N^{-\alpha} \ket{x}\right] \cdot \left[\bra{f(x)} Z_M^{\beta} \ket{f(x)}\right]
    \\ &\quad= \frac{1}{N}\sum_{x=0}^{N-1} \zeta_N^{-\alpha x}  \zeta_M^{\beta f(x)}
\end{split} \end{equation}
Above we used the trace cyclic property, distributivity of the tensor product, and that $Z_N^{-\alpha}$ is a diagonal matrix~\eqref{eq:powers_of_classical_error_basis_zn}.
\end{proof}

\begin{example}[Circuit tensor for classical identity and not]
To find the circuit tensor for the classical identity ($I_{classical}$) and not($X_{classical}$), we will directly use Proposition~\ref{proposition:clasical_ct}. Using both $N$ and $M$ equal to 2, we have $\zeta_N=\zeta_M=-1$:
\begin{equation}\begin{split}
    \circuittensor{I_{classical}}{Z^\alpha}{Z^\beta} &= \tfrac{1}{2}\sum_{x=0}^1 (-1)^{-\alpha x +\beta x}\\
    &=\begin{cases}
        1 &  \text{if~} \alpha=\beta,\\
        0 & otherwise.
    \end{cases}
\end{split}\end{equation}
    \begin{equation}\label{eq:ct_classical_identity}
        \circuitoperator{I_{classical}}  = e\errorbasis{I}{I}+e\errorbasis{Z}{Z}
    \end{equation}

\begin{equation}\begin{split}
    \circuittensor{X_{classical}}{Z^\alpha}{Z^\beta} &= \tfrac{1}{2}\sum_{x=0}^1 (-1)^{-\alpha x +\beta (1-x)}\\
    &=\begin{cases}
        (-1)^\alpha &  \text{if~} \alpha = \beta,\\
        0 & otherwise.
    \end{cases}
\end{split}\end{equation}

    \begin{equation}\label{eq:ct_classical_not}
        \circuitoperator{X_{classical}}  = e\errorbasis{I}{I} - e\errorbasis{Z}{Z}
    \end{equation}
\end{example}

\begin{proposition}[Circut tensor for a multi-input multi-output classical function]\label{proposition:multi_output_clasical_ct}
    For every classical function that has $n$ inputs and $m$ outputs $f:\left\{\{0,\dots, N_i-1\}\right\}_{i=1}^{n} \to \left\{\{0, \dots, M_j-1\}\right\}_{j=1}^{m}$, its circuit tensor is:
    \begin{align}
        &\circuittensor{f}{\otimes_{i=1}^{n}Z_{N_i}^{\alpha_i}}{\otimes_{j=1}^{m}Z_{M_j}^{\beta_j}}\nonumber\\
        &\qquad=\frac{1}{\prod_{i=1}^{n}N_i} \sum_{\mathbf{x}=(x_1, \dots, x_n)} \prod_{i=1}^n\zeta_{N_i}^{-\alpha_i x_i} \prod_{j=1}^m\zeta_{M_j}^{\beta_j f_{(j)}(\mathbf{x})}.
    \end{align}
    Above we used $f_{(j)}(\mathbf{x})$ as the $j$-th output of function f.
\end{proposition}

The proof follows the same lines as the proof of Proposition~\ref{proposition:clasical_ct}, we leave the details for the reader.

We finish this section with a circuit tensor construction for the xor function. Examples of other primitive boolean operations that commonly appear in quantum fault-tolerance circuits can be found in Appenndix~\ref{app:calssical_ct_ex}.

\begin{example}[Xor circuit tensor]\label{exp:xor}
    The operator for \textbf{xor} of two bits is:
    \begin{equation}
        A_{xor} = \sum_{x_0,x_1} \ket{x_0\oplus x_1}\bra{x_0,x_1} \in L((\mathbb{C}^2)^{\otimes 2},\mathbb{C}^2).
        \end{equation}
Using the circuit tensor definition in Proposition~\ref{proposition:multi_output_clasical_ct} we get:
\begin{align}
    \circuittensor{\mathtt{xor}}{Z^{\alpha_0}\otimes Z^{\alpha_1}}{Z^\beta} &= \tfrac{1}{4} \sum_{x_0,x_1} (-1)^{-\alpha_0 x_0 - \alpha_1 x_1 + \beta (x_0 \oplus x_1)} \nonumber\\
    &= \begin{cases}
     1 & \text{ if $\alpha_0 = \alpha_1 = \beta$,}\\
     0 & \text{otherwise.}
     \end{cases}
\end{align}
 Or using a tensor basis 
 \begin{equation}\label{eq:xor_ct}
 \circuitoperator{\mathtt{xor}} = e\errorbasis{I\otimes I}{I} + e\errorbasis{Z\otimes Z}{Z}.
 \end{equation} 
\end{example}

\section{Circuit tensor of a quantum channel}\label{section:ct_of_q_channel}

In this section, we extend the results of the previous section to quantum channels. The linear operator in Section~\ref{section:ct_linear_map} is replaced by a Kraus operator of the quantum channel, from which we build the Choi matrix (generalizing the Choi state) and the circuit tensor. Note that while the selection of Kraus operators for a channel is not unique, the Choi matrix and circuit tensor we form from them will be. The composition laws for Choi matrices and circuit tensors for channels are analogous to those of a single linear operator. We also provided detailed examples of projective and destructive measurements viewed as a quantum channel.

\begin{definition}[Choi matrix of a quantum channel]
    Let $\mathcal{A}:\mathfrak{H} \leadsto \mathfrak{K}$ be a quantum channel, and $\{A_j\} \subset L(\mathfrak{H},\mathfrak{K})$ a set of Kraus operators for this channel. Then the \emph{Choi matrix} of $\mathcal{A}$ is $T_{\mathcal{A}} = \sum_j \ketbra{T_{A_j}}{T_{A_j}}$.
\end{definition}

As the Kraus operators of a channel are not uniquely defined, it is not clear that the Choi matrix as defined is unique. Nonetheless, if as in the previous section we write each $\ket{T_{A_j}} = \sum_x \ket{x}\otimes A_j\ket{x}$ for some orthonormal basis $\{\ket{x}\}$ of $\mathfrak{H}$, we find
\begin{align}
    \nonumber T_\mathcal{A} &= \sum_j \left(\sum_x \ket{x}\otimes A_j\ket{x}\right)\left(\sum_y \bra{y}\otimes \bra{y}A_j^\dagger\right)\\\nonumber
    &= \sum_{j}\sum_{x,y} \ketbra{x}{y} \otimes A_j \ketbra{x}{y} A_j^\dagger\\
    &= \sum_{x,y} \ketbra{x}{y} \otimes \mathcal{M}(\ketbra{x}{y}),
\end{align}
and our Choi matrix coincides with the original formulation \cite{choi1975completely}, and so is independent of the choice of Kraus operators.

\begin{proposition}[Choi matrix for composition of quantum channels]\label{proposition:choi-matrix-composition}
    Let $\mathcal{A}:\mathfrak{H} \leadsto \mathfrak{K}$ and $\mathcal{B}:\mathfrak{K} \leadsto \mathfrak{L}$ be quantum channels. Then
    \begin{align}
        \nonumber &(I_\mathfrak{H} \otimes \bra\beta \otimes I_\mathfrak{L}) (T_{\mathcal{A}}\otimes T_{\mathcal{B}}) (I_\mathfrak{H}\otimes \ket\beta \otimes I_\mathfrak{L})\\
        &\qquad = \tfrac{1}{\dim(\mathfrak{K})^2} T_{\mathcal{B}\circ\mathcal{A}}
    \end{align}
\end{proposition}
\begin{proof}
    Let $\{A_j\} \subset L(\mathfrak{H},\mathfrak{K})$ and $\{B_k\} \subset L(\mathfrak{K},\mathfrak{L})$ be Kraus operators for $\mathcal{A}$ and $\mathcal{B}$ respectively. Then $\{B_kA_j\}$ is a set of Kraus operators for $\mathcal{B}\circ\mathcal{A}$, albeit not generally of minimal size. Nonetheless, from Proposition \ref{proposition:Choi:composition}
    \begin{equation}
        (I_\mathfrak{H} \otimes \bra\beta \otimes I_\mathfrak{L})(\ket{T_{A_j}}\otimes \ket{T_{B_k}}) = \tfrac{1}{\dim(\mathfrak{K})}  \ket{T_{B_kA_j}},
    \end{equation}
    and hence
    \begin{align}
        \nonumber&\tfrac{1}{\dim(\mathfrak{K})^2} T_{\mathcal{B}\circ\mathcal{A}} = \sum_{j,k} \ketbra{T_{B_kA_j}}{T_{B_kA_j}}\\
        &\quad = (I_\mathfrak{H} \otimes \bra\beta \otimes I_\mathfrak{L}) (T_{\mathcal{A}}\otimes T_{\mathcal{B}}) (I_\mathfrak{H}\otimes \ket\beta \otimes I_\mathfrak{L}).
    \end{align}
\end{proof}

The circuit tensor of a linear map was formed from the projector onto its Choi state, and hence the extension of circuit tensors to general quantum channels is natural.

\begin{definition}[Circuit tensor of a quantum channel]
    Let $\mathfrak{H},\mathfrak{K}$ be Hilbert spaces with error bases $\mathcal{E},\mathcal{E}'$ respectively, and let $\mathcal{A}:\mathfrak{H}\leadsto \mathfrak{K}$ be a quantum channel. Then the \emph{circuit tensor} of $\mathcal{A}$ is $\circuittensor{\mathcal{A}}{E}{E'} = \Tr((E^* \otimes E')T_{\mathcal{A}})$ where $E\in \mathcal{E}$, $E'\in \mathcal{E}'$, and $T_{\mathcal{A}}$ is the Choi matrix of $\mathcal{A}$.
\end{definition}

By construction, for a choice of Kraus operators $\{A_j\}$ for the channel $\mathcal{A}$ we have $T_{\mathcal{A}} = \sum_j \ketbra{T_{A_j}}{T_{A_j}}$, and hence
\begin{align}
    \nonumber \circuittensor{\mathcal{A}}{E}{E'} &= \sum_j \Tr((E^* \otimes E')\ketbra{T_{A_j}}{T_{A_j}})\\
    &= \sum_j \circuittensor{A_j}{E}{E'}.
\end{align}
That is, the circuit tensor of a quantum channel is just the sum of the circuit tensors of its constituent Kraus operators. As noted earlier, this would seem to indicate that the circuit tensor depends on the choice of Kraus operators. However as it is defined in terms of the Choi matrix, this is not the case.

Now, $\circuittensor{A_j}{E}{E'} = \frac{1}{\dim(\mathfrak{H})}\Tr(E^\dagger A_j^\dagger E' A_j)$, and so we can further write
\begin{equation}
    \circuitoperator{\mathcal{A}} = \tfrac{1}{\dim(\mathfrak{H})} \sum_{j,E,E'} \Tr(E^\dagger A_j^\dagger E' A_j) e^E_{E'}
\end{equation}
showing a further generalization of the $B$-tensor enumerator of \cite{cao2023quantum} to arbitrary quantum channels.

\begin{example}[Destructive measurement] Let $\{\ket{\phi_j}\}_{j=0}^{q-1}$ be an orthonormal basis of a Hilbert space $\mathfrak{H}$, and consider the operation of measuring a quantum state in this basis. That is, we input a quantum state from $\mathfrak{H}$ and output an element of $\{0,\dots,q-1\}$. Recall that as classical information the output is represented as density operators $\ketbra{0}{0}, \dots, \ketbra{q-1}{q-1}$ on $\mathbb{C}^q$. So as a quantum channel, this measurement operator is $\mathcal{MD}_{\{\ket{\phi_j}\}}: \mathfrak{H} \leadsto \mathbb{C}^q$. In particular,
\begin{equation}
    (\mathcal{MD}_{\{\ket{\phi_j}\}})(\rho) = \sum_{j=0}^{q-1} \ketbra{j}{j}\cdot \bra{\phi_j}\rho\ket{\phi_j}.
\end{equation}
Therefore we may take $\{\ket{j}\bra{\phi_j}\}_{j=0, \dots, q-1}$ as Kraus operators for this channel.

Suppose we have error basis $\mathcal{E}$ of $\mathfrak{H}$; for $\mathbb{C}^q$ we take error basis $\{Z^\alpha\}_{\alpha = 0}^{q-1}$. Then
\begin{align}
    \nonumber \circuittensor{\mathcal{MD}_{\{\ket{\phi_j}\}}}{E}{Z^\alpha} &= \frac{1}{q} \sum_{j=0}^{q-1} \Tr(E^\dagger \ket{\phi_j}\bra{j} Z^\alpha \ket{j}\bra{\phi_j})\\
    &= \frac{1}{q} \sum_{j=0}^{q-1}\zeta_q^{\alpha j} \bra{\phi_j} E^\dagger \ket{\phi_j},\label{eq:circuit-tensor-destructive-measurement}
\end{align}
where $\zeta_q$ is the primitive $q$-th root-of-unity. We see that owing to our choice of error basis of $\mathbb{C}^q$, the circuit tensor is encoding information about the measurement in the Fourier domain.
\end{example}

\begin{example}[Destructive Pauli measurement]
Let us specialize the above example to measuring a qubit with respect to a Pauli eigenbasis. We will consider $X$, with eigenbasis $\{\ket{+},\ket{-}\}$, the cases of $Y$ and $Z$ being similar. For compactness, let us simply write $\mathcal{MD}_X = \mathcal{MD}_{\{\ket{+},\ket{-}\}}$. Then by \eqref{eq:circuit-tensor-destructive-measurement},
\begin{equation}\label{eq:pauli-measurement-1}
    \circuittensor{\mathcal{MD}_{X}}{P}{Z^\alpha} = \tfrac{1}{2} \left( \bra{+} P \ket{+} + (-1)^\alpha  \bra{-} P \ket{-}\right).
\end{equation}

Now, we have $Y\ket{+} = -i\ket{-}$ and $Z\ket{+} = \ket{-}$ and so \eqref{eq:pauli-measurement-1}~reduces to
\begin{equation}
    \circuitoperator{\mathcal{MD}_{X}} = e\errorbasis{I}{I} + e\errorbasis{X}{Z}.
\end{equation}
Let us use the notation of $\mathcal{MD}_Z = \mathcal{MD}_{\{\ket{0},\ket{1}\}}$ and $\mathcal{MD}_Y = \mathcal{MD}_{\{\ket{i},\ket{-i}\}}$, then we can similarly get $\circuitoperator{\mathcal{MD}_{Y}} = e\errorbasis{I}{I} + e\errorbasis{Y}{Z}$ and $\circuitoperator{\mathcal{MD}_{Z}} = e\errorbasis{I}{I} + e\errorbasis{Z}{Z}$, and so in short
\begin{equation}\label{eq:md_tensor}
    \circuitoperator{\mathcal{MD}_{P}} = e\errorbasis{I}{I} + e\errorbasis{P}{Z}
\end{equation}
for any $P \in \{X, Y, Z\}$.
\end{example}

Let us now consider the slightly more complex example of non-demolition measurement. For simplicity, consider a binary projective measurement. As a quantum channel, this inputs a quantum state on the Hilbert space $\mathfrak{H}$ and outputs both a classical bit and the post-measurement state. Recall a classical bit is viewed as $\ketbra{0}{0}$ or $\ketbra{1}{1}$, density operators on $\mathbb{C}^2$. Hence the signature of this measurement as a quantum channel is $\mathfrak{H} \leadsto \mathbb{C}^2\otimes \mathfrak{H}$.

\begin{example}[Projective measurement]\label{example:binary-measurement}
    Let $\mathcal{MP}_\mathbf{R}:\mathfrak{H} \leadsto \mathbb{C}^2\otimes \mathfrak{H}$ be the quantum channel for measurement with respect to a binary projective-valued measure $\mathbf{R} = \{\Pi_0,\Pi_1\}$. We may write
    \begin{equation}
     (\mathcal{MP}_{\mathbf{R}})(\rho) = \ketbra{0}{0} \otimes \Pi_0\rho\Pi_0 + \ketbra{1}{1} \otimes \Pi_1\rho\Pi_1
    \end{equation}
    and so use as Kraus operators $\{\ket{0}\otimes \Pi_0, \ket{1}\otimes \Pi_1\}$. If $\mathcal{E}$ is our error basis for $\mathfrak{H}$, then for $\mathbb{C}^2\otimes \mathfrak{H}$ we take the error basis $\{ Z^\alpha\otimes E \::\: \alpha\in\{0,1\},\: E\in\mathcal{E}\}$. Then
    \begin{align}\label{eq:circuit-tensor-binary-measurement}
        &\circuittensor{\mathcal{MP}_\mathbf{R}}{E}{Z^\alpha\otimes E'} \nonumber\\&= \tfrac{1}{\dim(\mathfrak{H})} \sum_{j\in\{0,1\}} \Tr(E^\dagger (\bra{j}\otimes \Pi_j) (Z^\alpha\otimes E') (\ket{j}\otimes \Pi_j) \nonumber \\
        &= \tfrac{1}{\dim(\mathfrak{H})} \sum_{j\in\{0,1\}} (-1)^{\alpha j} \Tr(E^\dagger \Pi_j E' \Pi_j).
    \end{align}
\end{example}

\begin{example}[Projective Pauli measurement]\label{example:pauli-measurement}
    Let us specialize Example~\ref{example:binary-measurement} to an $n$-qubit Pauli operator measurement, $\Pi_j = \frac{1}{2}(I + (-1)^j S)$ where $S\in \mathcal{P}^n$. Write $\mathcal{MP}_S$ for this channel. Taking $\alpha = 0$ in \eqref{eq:circuit-tensor-binary-measurement},
    \begin{align}
        \circuittensor{\mathcal{MP}_S}{P}{I\otimes P'} &= \tfrac{1}{2^n}\left(\Tr(P\Pi_0 P'\Pi_0) +   \Tr(P\Pi_1 P'\Pi_1)\right)\nonumber\\
        &= \frac{1}{2^n} \begin{cases}
         \Tr(P^\dagger P') &\omega(P',S) = 1,\\
         0 &\text{otherwise,}
         \end{cases}
    \end{align}
    and for $\alpha = 1$ in \eqref{eq:circuit-tensor-binary-measurement},
    \begin{align}
        \circuittensor{\mathcal{MP}_S}{P}{Z\otimes P'} &= \tfrac{1}{2^n}\left(\Tr(P\Pi_0 P'\Pi_0) -   \Tr(P\Pi_1 P'\Pi_1)\right)\nonumber\\
        &= \frac{1}{2^n} \begin{cases}
        \Tr(P^\dagger P' S) &\omega(P',S) = 1,\\
        0 &\text{otherwise,}
        \end{cases}
    \end{align}
    In particular, we can write the circuit tensor as
    \begin{equation}\label{eq:ct_proj_meas_pauli_n_qubits}
        \circuitoperator{\mathcal{MP}_S} = \sum_{P\::\:\omega(P,S) = 1} e\errorbasis{P}{I\otimes P} + \mu(PS) e\errorbasis{P}{Z\otimes \mu(PS)PS}
    \end{equation}
    where as in Defenition~\ref{def:rel_pauli_phase}, we have used $\mu(PS)$ to represent the phase of $PS$ relative to the positve Pauli basis $\mathcal{P}^n$. Additionally, we write $\mu(PS)PS$ as the positive basis element.
\end{example}

\begin{theorem}[Circuit tensor for quantum channels composition]
    Let $\mathcal{A}:\mathfrak{H} \leadsto \mathfrak{K}$ and $\mathcal{B}:\mathfrak{K} \leadsto \mathfrak{L}$ be quantum channels. Then $ \circuitoperator{\mathcal{B}\circ\mathcal{A}} = \circuitoperator{\mathcal{A}}\circuitoperator{\mathcal{B}}$. That is,
    \begin{equation}
        \circuittensor{\mathcal{B}\circ\mathcal{A}}{E}{E'} = \sum_F \circuittensor{\mathcal{A}}{E}{F}\: \circuittensor{\mathcal{B}}{F}{E'}.
    \end{equation}
\end{theorem}
\begin{proof}
    We expand
    \begin{align}
        \nonumber \sum_F &\circuittensor{\mathcal{A}}{E}{F}\: \circuittensor{\mathcal{B}}{F}{E'} \\
        \nonumber &= \sum_F \Tr((E^* \otimes F)T_{\mathcal{A}}) \Tr((F^* \otimes E')T_{\mathcal{B}})\\
        \nonumber &= \sum_F \Tr\left( \left((E^* \otimes F)T_{\mathcal{A}}\right) \otimes \left((F^* \otimes E')T_{\mathcal{B}}\right)\right)\\
        &= \dim(\mathfrak{K})^2\cdot \Tr\left(\left(E^* \otimes \ket\beta\bra{\beta} \otimes E'\right)\left(T_{\mathcal{A}}\otimes T_{\mathcal{B}}\right)\right).
    \end{align}
    where we have used Lemma \ref{lemma:bell-state-stabilizer}. Now manipulating the trace,
    \begin{align}
        \nonumber &\Tr\left(\left(E^* \otimes \ket\beta\bra{\beta} \otimes E'\right)\left(T_{\mathcal{A}}\otimes T_{\mathcal{B}}\right)\right)\\
        \nonumber &\qquad= \Tr\left(\left(E^* \otimes I_{\mathfrak{K}} \otimes E'\right) \left( I_{\mathfrak{H}} \otimes \ket\beta \otimes I_{\mathfrak{K}} \right)\cdot\right.\\
        \nonumber &\qquad\qquad\qquad\left. (T_{\mathcal{A}}\otimes T_{\mathcal{B}})\left(I_{\mathfrak{H}} \otimes \bra\beta \otimes I_{\mathfrak{K}} \right)\right)\\
        \nonumber &\qquad= \Tr\left(\left(E^* \otimes E'\right) \left( I_{\mathfrak{H}} \otimes \bra\beta \otimes I_{\mathfrak{K}} \right)\right.\cdot\\
        \nonumber &\qquad\qquad\qquad\left. (T_{\mathcal{A}}\otimes T_{\mathcal{B}})\left(I_{\mathfrak{H}} \otimes \ket\beta \otimes I_{\mathfrak{K}} \right)\right)\\
        &\qquad= \tfrac{1}{\dim(\mathfrak{K})^2} \Tr\left( (E^* \otimes E') T_{\mathcal{B}\circ\mathcal{A}}\right),
    \end{align}
    where we have used Proposition \ref{proposition:choi-matrix-composition} above. Hence
    \begin{align}
        \sum_F \circuittensor{\mathcal{A}}{E}{F}\: \circuittensor{\mathcal{B}}{F}{E'} &= \Tr\left( (E^* \otimes E') T_{\mathcal{B}\circ\mathcal{A}}\right) \nonumber\\
        &= \circuittensor{\mathcal{B}\circ\mathcal{A}}{E}{E'}.
    \end{align}
\end{proof}

\section{Examples of quantum circuits}\label{section:ct_examples}
In this section, we provide constructions of circuit tensors for common quantum circuits that build upon the circuit tensors of quantum gates and measurements developed \S\S{\ref{section:ct_linear_map}~-~\ref{section:ct_of_q_channel}}. We start the section with an example showing how to construct a circuit tensor of a projective Pauli measurement, being one of the most common sub-circuits in fault tolerant quantum computing, using destructive measurements and Clifford operations. Next, we develop the circuit tensor for a classically controlled quantum gate, essential for post-measurement
feedback. We end the section with a large example, the teleportation circuit, that exhibits many of the basic blocks forming quantum circuits: state preparation, entangling gate, measurements, and classically controlled operations.

In Example~\ref{example:pauli-measurement} above, we derived the circuit tensor for a projective Pauli measurement,  $\mathcal{M}\mathcal{P}_S$ where $S \in \pauligroup^n$, by treating it as a general quantum channel. For our first example, we rederive this by composing the circuit tensors of a quantum circuit that implements the projective Pauli measurement with CNOTs, single-qubit Clifford gates, and a single-qubit destructive measurement. While projective Pauli measurement is native to, say, surface codes~\cite{litinski2019game}, other architectures do not have this as a native gate. Fig.~\ref{fig:projectiv_meas} shows circuits performing single qubit projective Z, X, and Y measurements.

\begin{figure}[ht]
    \centering
    \includegraphics[width=\linewidth]{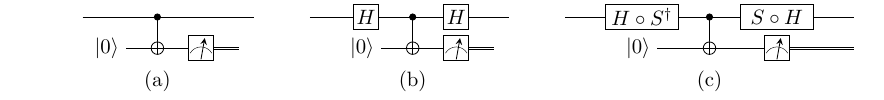}
    \caption[Projective Pauli measurements circuits]{Projective Pauli measurements: (a) projective Z measurement, (b) projective X measurement, (c) projective Y measurement.}
    \label{fig:projectiv_meas}
\end{figure}

In the simplest case of $S=Z$, and illustrated in Fig.~\ref{fig:projectiv_meas}(a), the circuit tensor of this circuit can be written, with a slight abuse of notation, as
\begin{equation}
    \circuitoperator{\mathcal{MP}_Z} = \circuitoperator{\mathcal{MD}^{(2)}_{Z}} \circ \circuitoperator{CX} \circ \circuitoperator{\mathcal{SP}^{(2)}_{\ket{0}}}.
\end{equation}
Here, the state preparation $\mathcal{SP}^{(2)}_{\ket{0}}$ occurs as the second qubit with an identity operation on the first qubit whose notation we suppress. Similarly, the destructive measurement $\mathcal{MD}^{(2)}_{Z}$ refers to the measurement of the second qubit, with the identity operation on the first.

In particular, from \eqref{eq:ct:state_prep} we have $\circuitoperator{\mathcal{SP}_{\ket{0}}} = e\errorbasis{}{I} + e\errorbasis{}{Z}$ and so:
\begin{align}
    \circuitoperator{\mathcal{SP}^{(2)}_{\ket{0}}} &= \circuitoperator{I} \otimes \circuitoperator{\mathcal{SP}_{\ket{0}}}\nonumber\\ 
    &= e\errorbasis{I}{I\otimes I}  + e\errorbasis{X}{X\otimes I} + e\errorbasis{Y}{Y\otimes I} + e\errorbasis{Z}{Z\otimes I} \nonumber \\
    &\qquad + e\errorbasis{I}{I\otimes Z}  + e\errorbasis{X}{X\otimes Z} + e\errorbasis{Y}{Y\otimes Z} + e\errorbasis{Z}{Z\otimes Z}.
\end{align}
Then from \eqref{eq:ct_cnot}, the application of the CNOT yields
\begin{align}
    &\circuitoperator{CX} \circ \circuitoperator{\mathcal{SP}^{(2)}_{\ket{0}}} = e\errorbasis{I}{I\otimes I}  + e\errorbasis{X}{X\otimes X} + e\errorbasis{Y}{Y\otimes X}  + e\errorbasis{Z}{Z\otimes I} \nonumber \\
    &\qquad+ e\errorbasis{I}{Z\otimes Z}  - e\errorbasis{X}{Y\otimes Y} + e\errorbasis{Y}{X\otimes Y} + e\errorbasis{Z}{I\otimes Z}.\label{eq:example-77}
\end{align}
Now, from \eqref{eq:md_tensor}, $\circuitoperator{\mathcal{MD}_{Z}} = e\errorbasis{I}{I} + e\errorbasis{Z}{Z}$ and
\begin{align}
    \circuitoperator{\mathcal{MD}^{(2)}_{Z}} &= \circuitoperator{I} \otimes \circuitoperator{\mathcal{MD}_{Z}}\nonumber\\
    &= e\errorbasis{I\otimes I}{I\otimes I}  + e\errorbasis{X\otimes I}{X\otimes I} + e\errorbasis{Y\otimes I}{Y\otimes I} + e\errorbasis{Z\otimes I}{Z\otimes I} \nonumber \\
    &\qquad + e\errorbasis{I\otimes Z}{I\otimes Z}  + e\errorbasis{X\otimes Z}{X\otimes Z} + e\errorbasis{Y\otimes Z}{Y\otimes Z} + e\errorbasis{Z\otimes Z}{Z\otimes Z}.\label{eq:example-78}
\end{align}

Composing \eqref{eq:example-77} with \eqref{eq:example-78} gives
\begin{equation} \label{eq:simple_projective_z_mes}
    \circuitoperator{\mathcal{MP}_Z} = e\errorbasis{I}{I\otimes I} + e\errorbasis{Z}{Z \otimes I} + e\errorbasis{I}{Z\otimes Z} + e\errorbasis{Z}{I\otimes Z},
\end{equation}
which aligns with \eqref{eq:ct_proj_meas_pauli_n_qubits}.

Only slightly more complicated are the $X$-basis and $Y$-basis projective measurement 
 of Fig.~\ref{fig:projectiv_meas}(b,c). Here we need to add Clifford operations to our first qubit before and after the CNOT. We leave the full derivation for the reader, referring to (\ref{eq:ct_H}, \ref{eq:ct_S_H}, \ref{eq:ct_H_Sd}):
\begin{align} \label{eq:simple_projective_x_mes}
    \circuitoperator{\mathcal{MP}_X} &= \left[\circuitoperator{H} \otimes \circuitoperator{\mathcal{MD}_{Z}}\right] \circ \circuitoperator{CX} \nonumber \\
    &\quad\qquad \circ \left[\circuitoperator{H} \otimes \circuitoperator{\mathcal{SP}_{\ket{0}}}\right] \nonumber\\
    &\quad=e\errorbasis{I}{I\otimes I}  + e\errorbasis{X}{X\otimes I} + e\errorbasis{I}{X\otimes Z}  + e\errorbasis{X}{I\otimes Z},
\end{align}
and
\begin{align} \label{eq:simple_projective_y_mes}
    \circuitoperator{\mathcal{MP}_Y} &= \left[\circuitoperator{S\circ H} \otimes \circuitoperator{\mathcal{MD}_{Z}}\right] \circ \circuitoperator{CX} \nonumber \\
    &\quad\qquad \circ \left[\circuitoperator{H\circ S^\dagger} \otimes \circuitoperator{\mathcal{SP}_{\ket{0}}}\right] \nonumber\\
    &\quad=e\errorbasis{I}{I\otimes I}  + e\errorbasis{Y}{Y\otimes I} + e\errorbasis{I}{Y\otimes Z}  + e\errorbasis{X}{I\otimes Z}.
\end{align}

From the above, it is easy to see how to construct a general projective Pauli measurement. To that end, we will note the pre and post-entanglement operations as follows:
\begin{align}
    &\circuitoperator{\textit{pre-en}(P)} =  
      \begin{cases}
      \circuitoperator{I} & P\in \{I,Z\}, \\
      \circuitoperator{H} & P=X, \\
      \circuitoperator{\mathcal{S\circ H}} & P=Y, \\
      \end{cases} \\
    &\circuitoperator{\textit{post-en}(P)} = 
    \begin{cases}
      \circuitoperator{I} & P\in \{I,Z\}, \\
      \circuitoperator{H} & P=X, \\
      \circuitoperator{\mathcal{H\circ S^\dagger}} & P=Y. \\
      \end{cases} 
\end{align}

We will also note the entanglement operation between qubit $i$ and $j$, while all the other qubits are idle as:

\begin{equation}
    \circuitoperator{\textit{entangle}^{(i,j)}(P)} =  
    \begin{cases}
      \circuitoperator{I} \otimes \circuitoperator{I} & P=I, \\
      \circuitoperator{CX} & P \neq I. \\
    \end{cases}
\end{equation}

Then for any Pauli operator $P\in\mathcal{P}^n$ we have:
\begin{align}
&\circuitoperator{\mathcal{MP}_P} \nonumber \\
&\quad= \circuitoperator{\mathcal{MD}^{(n+1)}_{Z}} \circ \left[\bigotimes_{i=1}^{n} \circuitoperator{\textit{post-en}(P_i)} \otimes \circuitoperator{I}\right] \nonumber \\
& \qquad \circ \bigcirc_{i=1}^n \circuitoperator{\textit{entangle}^{(i,n+1)}(P_i)}  \nonumber \\
& \qquad \circ \left[ \bigotimes_{i=1}^{n} \circuitoperator{\textit{pre-en}(P_i)} \otimes \circuitoperator{I}\right]  \circ \circuitoperator{\mathcal{SP}^{(n+1)}_{\ket{0}}} .
\end{align}

For our second example we consider post-measurement feedback, a common construction in quantum algorithms and protocols: perform one of a selection of operations based on the outcome of a measurement. Some examples include gate injection~\cite[\S{10.6.2}]{Nielsen_Chuang_2010}; repeat-until-success circuits \cite{paetznick2013repeat, wiebe2014quantum}; and quantum state teleportation~\cite[\S{1.3.7}]{Nielsen_Chuang_2010}. The last of these we will treat in detail below. 

At this stage we only consider selecting between one of two operations; we will consider the general case in \S\ref{section:noise-analysis} below. To keep a high level of generality, let us suppose each of our operations are quantum channels with signature $\mathfrak{H} \leadsto \mathfrak{K}$. Our selector bit provides an addition classical input, whose associated Hilbert space is $\mathbb{C}^2$, hence our overall circuit will be a quantum channel with signature $\mathbb{C}^2\otimes \mathfrak{H} \leadsto \mathfrak{K}$.

\begin{proposition}[Classical selection between two quantum channels]\label{proposition:selector-channel}
Let $\mathcal{M}^{(0)}, \mathcal{M}^{(1)}: \mathfrak{H} \leadsto \mathfrak{K}$ be quantum channels, and let $\mathcal{M}:\mathbb{C}^2\otimes \mathfrak{H} \leadsto \mathfrak{K}$ be the channel that selects $\mathcal{M}^{(0)},\mathcal{M}^{(1)}$ given the input bit. Then the circuit tensor of $\mathcal{M}$ is:
\begin{equation}
    \circuittensor{\mathcal{M}}{Z^\alpha\otimes E}{E'} = \frac{1}{2}\left(\circuittensor{\mathcal{M}^{(0)}}{E}{E'} + (-1)^\alpha \circuittensor{\mathcal{M}^{(1)}}{E}{E'}\right).
\end{equation}
\end{proposition}
\begin{proof}
    Let $\{A^{(0)}_j\}$ and $\{A^{(1)}_k\}$ be sets of Kraus operators for $\mathcal{M}^{(0)}$ and $\mathcal{M}^{(1)}$ respectively. Then $\left\{\bra{0}\otimes A^{(0)}_j, \bra{1}\otimes A^{(1)}_k\right\}$ is a set of Kraus operators for $\mathcal{M}$. Now, for example, 
    \begin{align}
        \nonumber \ket{T_{\bra{0}\otimes A^{(0)}_j}} &= \tfrac{1}{\sqrt{2\dim\mathfrak{H}}} \sum_{b,x} \ket{b,x} \otimes \left(\bra{0}\otimes A^{(0)}_j\right) \ket{b,x}\\
        &= \tfrac{1}{\sqrt{2}} \ket{0}\otimes \ket{T_{A^{(0)}_j}}.
    \end{align}
    and so
    \begin{equation}\label{eq:Choi-matrix-sum}
        T_{\mathcal{M}} = \tfrac{1}{2}\left( \ketbra{0}{0} \otimes T_{\mathcal{M}^{(0)}} + \ketbra{1}{1}\otimes T_{\mathcal{M}^{(1)}}\right) 
    \end{equation}
    
    Therefore, we compute:
    \begin{align}
        \nonumber &\circuittensor{\mathcal{M}}{Z^\alpha\otimes E}{E'} = \Tr((Z^\alpha \otimes E^*\otimes E)     T_{\mathcal{M}})\\
        \nonumber &\quad = \frac{1}{2}(\Tr((Z^\alpha \otimes E^*\otimes E) (\ketbra00 \otimes T_{\mathcal{M}^{(0)}}))\\
        \nonumber &\qquad\qquad + \Tr((Z^\alpha \otimes E^*\otimes E) (\ketbra11 \otimes T_{\mathcal{M}^{(1)}})))\\
        \nonumber &\quad = \frac{1}{2}\left(\Tr((E^*\otimes E) T_{\mathcal{M}^{(0)}}) + (-1)^\alpha \Tr((E^*\otimes E)  T_{\mathcal{M}^{(1)}})\right)\\
        &\quad = \frac{1}{2}\left(\circuittensor{\mathcal{M}^{(0)}}{E}{E'} + (-1)^\alpha \circuittensor{\mathcal{M}^{(1)}}{E}{E'}\right).
    \end{align}
\end{proof}

Next, we will show a simple example of a classically controlled Pauli operation.

\begin{example}[Classically controlled Pauli operation]\label{ex:clas_cp}
    In a classically controlled $n$-qubit Pauli operation, we have $n$ qubits that go through a quantum channel controlled by a classical wire. The two possible channels are the $n$-qubit identity with the circuit tensor $\circuittensor{\mathcal{I}}{E}{E'}=\tfrac{1}{2^n}\Tr(E^\dagger E')$ or the $n$-qubit Pauli $P$ channel, with the circuit tensor $\circuittensor{P}{E}{E'}=\tfrac{1}{2^n}\Tr(E^\dagger P^\dagger E' P)$. Therefore, using Proposition~\ref{proposition:selector-channel} the circuit tensor of the classically controlled $n$-qubit $P$ operations is: 
    \begin{align}
        \circuittensor{\text{cntl-}P}{Z^\alpha \otimes E}{E'} &= \tfrac{1}{2^{n+1}}\left[\Tr(E^\dagger E') + (-1)^\alpha \Tr(E^\dagger P^\dagger E' P)\right] \nonumber \\
        &= \begin{cases}
                \tfrac{1}{2}\left(1+(-1)^\alpha\omega(E,P)\right) &  E = E'  \\
                0 & E \ne E'.
            \end{cases}
    \end{align}
    And using the tensor basis, we get
    \begin{equation}\label{eq:classical_cont_pauli_ct}
    \circuitoperator{\text{cntl-}P} = \sum_{\text{ \tiny$ \begin{matrix} E\in \mathcal{P}^n \\ \omega(E,P) = 1   \end{matrix}$}} e\errorbasis{I\otimes E}{E} + \sum_{\text{ \tiny$ \begin{matrix} E\in \mathcal{P}^n \\ \omega(E,P) = -1   \end{matrix}$}} e\errorbasis{Z\otimes E}{E}  \ .
    \end{equation}
    
    We specify the Pauli to be $Z$ or $X$ for example, then in the component-free form we get:
    \begin{equation}\label{eq:ct_clas_cz}
        \circuitoperator{\text{cntl-}Z} = e\errorbasis{I\otimes I}{I} + e\errorbasis{Z\otimes X}{X} + e\errorbasis{Z\otimes Y}{Y} +e\errorbasis{I\otimes Z}{Z} \ ,
    \end{equation}
    \begin{equation}\label{eq:ct_clas_cx}
        \circuitoperator{\text{cntl-}X} = e\errorbasis{I\otimes I}{I} + e\errorbasis{I\otimes X}{X} + e\errorbasis{Z\otimes Y}{Y} +e\errorbasis{Z\otimes Z}{Z} \ .
    \end{equation}
\end{example}

We will conclude this section with a construction of the circuit tensor for the teleportation circuit. It consists of many building blocks we developed in the paper, and can act as a cookbook on how to construct circuit tensors for complicated circuits.

\begin{figure}[ht]
    \centering
    \includegraphics[width=0.9\linewidth]{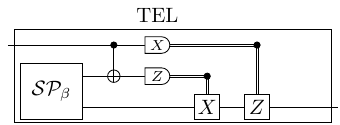}
    \caption[Teleportion circuit]{Teleportion circuit composed of a Bell-state preparation (including the creation of two new qubits - $\mathcal{S}\mathcal{P}_{\beta}$), an entangling operation, X and Z destructive measurements, and classically controlled Pauli corrections for the output qubit. In the paper we refer to the full circuit as TEL.}
    \label{fig:teleportation}
\end{figure}

\begin{example}[Circuit tensor for the teleportation circuit]\label{ex:teleportation}
We will show how to construct the circuit tensor of the teleportation circuit. Fig.~\ref{fig:teleportation} presents the teleportation circuit, it has an input of one qubit and an output of one qubit. Internally it creates a Bell State~\eqref{eq:ct_bell_sp}, entangels the input qubit to one half of the Bell State~\eqref{eq:ct_cnot}, then it performs the destructive $X$ and $Z$ measurements~\eqref{eq:md_tensor} and applies the classically controlled Pauli corrections~\eqref{eq:ct_clas_cz}~\eqref{eq:ct_clas_cx}. To create the teleportation circuit tensor We also use the circuit tensor of the quantum and classical identity~\eqref{eq:ct_identity}~\eqref{eq:ct_classical_identity}.

\begin{equation}\label{eq:noiseless_teleportation}\begin{split}
 \circuitoperator{TEL} &=\circuitoperator{\text{cntl-}Z} \circ [\circuitoperator{I_{classical}} \otimes \circuitoperator{\text{cntl-}X}]  \\
 & \circ [\circuitoperator{\mathcal{MD}_X} \otimes \circuitoperator{\mathcal{MD}_Z} \otimes \circuitoperator{I} ]  \\
 &  \circ [\circuitoperator{CX} \otimes \circuitoperator{I}] \circ [\circuitoperator{I} \otimes \circuitoperator{\mathcal{SP}_\beta}] \ .
\end{split}\end{equation}

For ease of presentation, we will derive the above in parts:

\begin{equation} \label{eq:tel_p2} \begin{split}
    &\circuitoperator{\mathcal{MD}_X} \otimes \circuitoperator{\mathcal{MD}_Z} \otimes \circuitoperator{I} \\
    &\quad= (e\errorbasis{I}{I} + e\errorbasis{X}{Z}) \otimes (e\errorbasis{I}{I} + e\errorbasis{Z}{Z}) \otimes \circuitoperator{I} \\
    &\quad =\quad e\errorbasis{I\otimes I\otimes I}{I\otimes I\otimes I} + ~e\errorbasis{X\otimes I\otimes I}{Z\otimes I\otimes I} + e\errorbasis{I\otimes Z\otimes I}{I\otimes Z\otimes I} ~+ e\errorbasis{X\otimes Z\otimes I}{Z\otimes Z\otimes I} \\
    &\qquad + e\errorbasis{I\otimes I\otimes X}{I\otimes I\otimes X} + e\errorbasis{X\otimes I\otimes X}{Z\otimes I\otimes X} + e\errorbasis{I\otimes Z\otimes X}{I\otimes Z\otimes X} + e\errorbasis{X\otimes Z\otimes X}{Z\otimes Z\otimes X} \\
    &\qquad + e\errorbasis{I\otimes I\otimes Y}{I\otimes I\otimes Y} + e\errorbasis{X\otimes I\otimes Y}{Z\otimes I\otimes Y} + e\errorbasis{I\otimes Z\otimes Y}{I\otimes Z\otimes Y} ~+ e\errorbasis{X\otimes Z\otimes Y}{Z\otimes Z\otimes Y} \\
    &\qquad + e\errorbasis{I\otimes I\otimes Z}{I\otimes I\otimes Z} + e\errorbasis{X\otimes I\otimes Z}{Z\otimes I\otimes Z} + e\errorbasis{I\otimes Z\otimes Z}{I\otimes Z\otimes Z} ~+ e\errorbasis{Z\otimes X\otimes Z}{Z\otimes Z\otimes Z}
\end{split} \end{equation}

\begin{equation}\label{eq:tel_p1}\begin{split}
    &\circuitoperator{\text{cntl-}Z} \circ [\circuitoperator{I_{classical}} \otimes \circuitoperator{\text{cntl-}X}]  \\ 
    &~= (e\errorbasis{I\otimes I}{I} + e\errorbasis{Z\otimes X}{X} + e\errorbasis{Z\otimes Y}{Y} +e\errorbasis{I\otimes Z}{Z}) \\
    &\quad \circ (\quad e\errorbasis{I\otimes I\otimes I}{I\otimes I} + e\errorbasis{I\otimes I\otimes X}{I\otimes X} + e\errorbasis{I\otimes Z\otimes Y}{I\otimes Y} +e\errorbasis{I\otimes Z\otimes Z}{I\otimes Z} \\
    &\qquad + e\errorbasis{Z\otimes I\otimes I}{Z\otimes I} + e\errorbasis{Z\otimes I\otimes X}{Z\otimes X} + e\errorbasis{Z\otimes Z\otimes Y}{Z\otimes Y} +e\errorbasis{Z\otimes Z\otimes Z}{Z\otimes Z}) \\
    &~= e\errorbasis{I\otimes I\otimes I}{I} +e\errorbasis{I\otimes Z\otimes Z}{Z} +  e\errorbasis{Z\otimes I\otimes X}{X} + e\errorbasis{Z\otimes Z\otimes Y}{Y}
\end{split} \end{equation}

\begin{equation} \label{eq:tel_p3}\begin{split}
&[\circuitoperator{CX} \otimes \circuitoperator{I}] \circ [ \circuitoperator{I} \otimes \circuitoperator{\mathcal{SP}_\beta}] \\
&\quad=[\circuitoperator{CX} \otimes \circuitoperator{I}]\\
&\qquad \circ(\quad e\errorbasis{I}{I\otimes I\otimes I} + e\errorbasis{I}{I\otimes X\otimes X} ~- e\errorbasis{I}{I\otimes Y\otimes Y} ~+ e\errorbasis{I}{I\otimes Z\otimes Z} \\
&\qquad\quad +e\errorbasis{X}{X\otimes I\otimes I}  + e\errorbasis{X}{X\otimes X\otimes X} - e\errorbasis{X}{X\otimes Y\otimes Y}  + e\errorbasis{X}{X\otimes Z\otimes Z} \\
&\qquad\quad + e\errorbasis{Y}{Y\otimes I\otimes I} + e\errorbasis{Y}{Y\otimes X\otimes X} - e\errorbasis{Y}{Y\otimes Y\otimes Y} + e\errorbasis{Y}{Y\otimes Z\otimes Z} \\
&\qquad\quad + e\errorbasis{Z}{Z\otimes I\otimes I} + e\errorbasis{Z}{Z\otimes X\otimes X}- e\errorbasis{Z}{Z\otimes Y\otimes Y} + e\errorbasis{Z}{Z\otimes Z\otimes Z}) \\
&\quad= e\errorbasis{I}{I\otimes I\otimes I} + e\errorbasis{I}{I\otimes X\otimes X} - e\errorbasis{I}{Z\otimes Y\otimes Y} + e\errorbasis{I}{Z\otimes Z\otimes Z} \\
&\qquad + e\errorbasis{X}{X\otimes X\otimes I} + e\errorbasis{X}{X\otimes I\otimes X} - e\errorbasis{X}{Y\otimes Z\otimes Y} - e\errorbasis{X}{Y\otimes Y\otimes Z}\\
&\qquad  + e\errorbasis{Y}{Y\otimes X\otimes I} + e\errorbasis{Y}{Y\otimes I\otimes X} + e\errorbasis{Y}{X\otimes Z\otimes Y} + e\errorbasis{Y}{X\otimes Y\otimes Z}  \\
&\qquad + e\errorbasis{Z}{Z\otimes I\otimes I} + e\errorbasis{Z}{Z\otimes X\otimes X} - e\errorbasis{Z}{I\otimes Y\otimes Y} + e\errorbasis{Z}{I\otimes Z\otimes Z}
\end{split} \end{equation}

Composing \eqref{eq:tel_p1} with \eqref{eq:tel_p2} we get:
\begin{equation}\label{eq:tel_p1_2}
    e\errorbasis{I\otimes I\otimes I}{I} +e\errorbasis{I\otimes Z\otimes Z}{Z} +  e\errorbasis{X\otimes I\otimes X}{X} +   e\errorbasis{X\otimes Z\otimes Y}{Y} \ .
\end{equation}

Further composing \eqref{eq:tel_p1_2} and \eqref{eq:tel_p3} results with:
\begin{equation}
\circuitoperator{TEL} = e\errorbasis{I}{I} + e\errorbasis{X}{X} + e\errorbasis{Y}{Y} + e\errorbasis{Z}{Z} = \circuitoperator{I}.
\end{equation}

\end{example}

As expected, we got the identity circuit tensor~\eqref{eq:ct_identity}. This might seem meaningless, but next, we will show how to add noise sources to circuit tensors, enabling an easy way to analyze a noisy circuit.

\section{Noise Analysis in quantum circuits}\label{section:noise-analysis}

In this section, we show how to incorporate noise sources into circuits and analyze these using the circuit tensor. Here, we will only treat models that depend on a finite number of noise modes, which we index as $m=0, \dots, M-1$ (where implicitly, mode $m=0$ denotes no error). Given a quantum channel for each noise mode, we form the composite quantum channel that takes $m$ as a classical input selecting the error mode. We then define a ``trace'' akin to \cite[Lemma IV.1]{cao2023quantum} that attaches a formal variable to this classical input, which may also be interpreted as a weight of the corresponding error mode. This defines a ``circuit enumerator'' that we will use extensively in later sections to count error paths in syndrome extraction circuits of quantum codes.

Setting notation, for a given quantum circuit let us write $\mathcal{U}: \mathfrak{H} \leadsto \mathfrak{K}$ for the quantum channel of that circuit with no errors. We then write $\widetilde{\mathcal{U}}(m):\mathfrak{H} \leadsto \mathfrak{K}$ for the analogous noisy circuit operating under noise mode $m$ (so $\widetilde{\mathcal{U}}(0) = \mathcal{U}$).  We combine these channels into a channel, $\widetilde{\mathcal{U}}: \mathbb{C}^M \otimes \mathfrak{H} \leadsto \mathfrak{K}$, as illustrated in Fig.~\ref{fig:muxed-error-channel}a,  where the first factor in the domain is the (classical) selection for the noise mode. Formally,
\begin{equation}\label{eq:composite-error-channel}
\widetilde{\mathcal{U}}(\ketbra{m}{m}\otimes \rho) = \widetilde{\mathcal{U}}(m)(\rho).
\end{equation}

\begin{figure}[ht]
    \centering
    \includegraphics{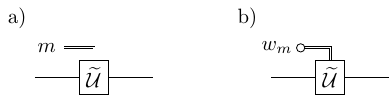}
    \caption{(a) An error channel accepts a classical input to select the error mode. (b) The trace (represented by the open circle) creates the circuit enumerator of the channel, whose variables track occurrences of error modes.}
    \label{fig:muxed-error-channel}
\end{figure}

\begin{lemma}\label{lemma:circuit-tensor-error-model}
    $\circuitoperator{\widetilde{\mathcal{U}}} = \sum_{m,E,E'} \zeta_M^{-\alpha m}
    \circuittensor{\widetilde{\mathcal{U}}(m)}{E}{E'}\: e\errorbasis{Z^\alpha\otimes E}{E'}$, where $\zeta_M = e^{2\pi i/M}$ is the canonical $M$-th root-of-unity.
\end{lemma}
\begin{proof}
    Just as in the Proposition~\ref{proposition:selector-channel} above, the Choi matrix of the composite channel has 
    \begin{equation}
        T_{\widetilde{\mathcal{U}}} = \frac{1}{M}\sum_{m=0}^{M-1} \ketbra{m}{m}\otimes T_{\widetilde{\mathcal{U}}(m)}.
    \end{equation} 
    Thus by definition
    \begin{align}
        \nonumber \circuittensor{\widetilde{\mathcal{U}}}{Z^\alpha\otimes E}{E'} &= \Tr((Z^{-\alpha}\otimes E^* \otimes E') T_{\widetilde{\mathcal{U}}})\\
        \nonumber &= \sum_{m=0}^{M-1} \bra{m}Z^{-\alpha}\ket{m} \Tr((E^* \otimes E')T_{\widetilde{\mathcal{U}}(m)})\\
        &= \sum_{m=0}^{M-1} \zeta_M^{-\alpha m} \circuittensor{\widetilde{\mathcal{U}}(m)}{E}{E'} \ .
    \end{align}
\end{proof}

Note that this single channel is merely an encoding of all the different error modes of the circuit; there is yet no assignment of weight or likelihood attached to each mode. As the factor $\mathbb{C}^M$ supports only classical information, a density operator $\sigma \in \mathcal{S}(\mathbb{C}^M)$ must be of the form $\sigma = \sum_{m=0}^{M-1} w_m \ketbra{m}{m}$. So $\widetilde{\mathcal{U}}(\sigma \otimes \rho)$ is the result of the noisy channel on $\rho \in \mathcal{S}(\mathfrak{H})$ where each error mode $m$ is selected with probability $w_m$.

We want to treat the $w_m$ as formal variables, and so create a \emph{``circuit enumerator''} that tracks the occurrence of each error mode. We can achieve that with a trace-like operation $\widetilde{\Tr}: \mathcal{S}(\mathbb{C}^M) \to \mathbb{R}[w_0,\dots, w_{M-1}]$, as illustrated in Fig.~\ref{fig:muxed-error-channel}b, by $\widetilde{\Tr}(\ketbra{m}{m}) = w_m$ extended linearly to all density operators--compare to \cite[Definition IV.2]{cao2023quantum}.

But, at the level of circuit tensors, we instead define Fourier variables $u_\alpha = \sum_{m=0}^{M-1} \zeta_M^{\alpha m}w_m$, where, as above $\zeta_M = e^{2\pi i/M}$, and extend this trace to tensors, for which we do not introduce new notation, as 
\begin{equation}\label{eq:tensor-trace-definition}
    \widetilde{\Tr}(e\errorbasis{Z^\alpha\otimes E}{E'}) = u_\alpha e\errorbasis{E}{E'} \ .
\end{equation}
The link between the two above notions of trace is given as follows.

\begin{proposition}\label{proposition:circuit-enumerator}
    Let $\widetilde{\mathcal{U}}(m):\mathfrak{H} \leadsto \mathfrak{K}$ be quantum channels for $m=0, \dots, M-1$, each with an asssociated variable $w_m$. Define $\widetilde{\mathcal{U}}: \mathbb{C}^M \otimes \mathfrak{H} \leadsto \mathfrak{K}$ as in (\ref{eq:composite-error-channel}).
    If as above, $\widetilde{\Tr}(e\errorbasis{Z^\alpha\otimes E}{E'}) = u_\alpha e\errorbasis{E}{E'}$ with $u_\alpha = \sum_{m} \zeta_M^{\alpha m}w_m$ we then have
    \begin{equation}\label{eq:circuit-enumerator}
        \widetilde{\Tr}\left[\circuitoperator{\widetilde{\mathcal{U}}}\right] = \sum_{m=0}^{M-1} w_m \circuitoperator{\widetilde{\mathcal{U}}(m)}.
    \end{equation}
\end{proposition}
\begin{proof}
    We apply the trace definition (\ref{eq:tensor-trace-definition}) to the circuit tensor as given in Lemma~\ref{lemma:circuit-tensor-error-model}, yielding
    \begin{align}
        \nonumber \widetilde{\Tr}\left[\circuitoperator{\widetilde{\mathcal{U}}}\right] &= \sum_{\alpha,m,E,E'} \frac{1}{M} \zeta_M^{-\alpha m} u_\alpha \circuittensor{\widetilde{\mathcal{U}}(m)}{E}{E'}\: e^E_{E'}\\
        \nonumber &= \sum_{m,E,E'} w_m \circuittensor{\widetilde{\mathcal{U}}(m)}{E}{E'}\: e^E_{E'}\\
        &= \sum_m w_m \circuitoperator{\widetilde{\mathcal{U}}(m)}.
    \end{align}
\end{proof}

We can consider (\ref{eq:circuit-enumerator}) as a definition for the circuit enumerator of the noisy circuit $\widetilde{\mathcal{U}}$. However, to utilize the machinery developed in \cite{cao2023quantum} we instead attach a weight function to an error model. Then the circuit enumerator is constructed analogously to the tensor enumerators of that work.

\begin{definition}[Circuit enumerator]
    Let $\mathcal{E}$ be an error basis. A weight function is any function $\mathrm{wt}:\mathcal{E} \to \mathbb{Z}^k_{\geq 0}$. If $u = (u_0, \dots, u_k)$ is a tuple of indeterminates, we write $u^{\mathrm{wt}(E)} = u_0^{\mathrm{wt}(E)_0}\cdots  u_{k}^{\mathrm{wt}(E)_k}$. The circuit enumerator of an error channel $\widetilde{\mathcal{U}}$ (as in \eqref{eq:composite-error-channel}) with associated weight function $\mathrm{wt}$ is defined by
    \begin{equation}\label{eq:circuit-enumerator-definition}
        \widetilde{\Tr}\left[\circuitoperator{\widetilde{\mathcal{U}}}\right] = \sum_{\alpha,E,E'} \circuittensor{\widetilde{\mathcal{U}}}{Z^\alpha\otimes E}{E'} u^{\mathrm{wt}(Z^\alpha)} e\errorbasis{E}{E'}.
    \end{equation}
\end{definition}

We hasten to note that the trace on the left side of (\ref{eq:circuit-enumerator-definition}) is purely notational, but from (\ref{eq:tensor-trace-definition}) it is consistent with the trace used in the previous Proposition. We also note that our definition of a weight function differs from that in \cite{cao2023quantum}; using the terminology of that paper, we only have need of scalar weight functions and so restrict to that case.

\begin{example}[Bit flip error]
We begin with a very simple classical error model: the bit flip error. We view the bit flip as a classical xor operation of an input bit with the bit that controls whether the error occurs, as illustrated in Fig.~\ref{fig:readout_circ}. In Example~\ref{exp:xor} we found $\circuitoperator{\mathtt{xor}} = e\errorbasis{I\otimes I}{I} + e\errorbasis{Z\otimes Z}{Z}$, and so performing the trace (\ref{eq:tensor-trace-definition}) we find
\begin{equation}
    \widetilde{\Tr}\left[\circuitoperator{\mathtt{xor}}\right] = u_0 e\errorbasis{I}{I} + u_1 e\errorbasis{Z}{Z}.
\end{equation}
From this equation, we can deduce the weight function of our error model: $\mathrm{wt}(I) = (1,0)$ and $\mathrm{wt}(Z) = (0,1)$. For purposes of enumerating bit-flip errors, we take weight variables $w_0 = 1$ (no bit-flip) and $w_1 = r$ (one bit-flip), so that $r$ becomes the variable in the associated weight enumerator. Using the relationship $u_0 = w_0 + w_1$ and $u_1 = w_0 - w_1$ we have
\begin{equation}\label{eq:circuit-xor-traced}
    \widetilde{\Tr}[\circuitoperator{\mathtt{xor}}] = (1+r) e\errorbasis{I}{I} + (1 - r) e\errorbasis{Z}{Z}.
\end{equation}
\end{example}

\begin{figure}[t]
    \centering
    \includegraphics{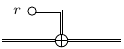}
    \caption{Modeling a bit-flip on a (classical) measurement result. Here, $r$ is its weight enumerator variable.}
    \label{fig:readout_circ}
\end{figure}

\begin{example}[Pauli readout assignment error]
    In Example~\ref{example:pauli-measurement}, we found the circuit tensor for the projective measurement of a Pauli operator $S \in \mathcal{P}^n$ is given by $\circuitoperator{\mathcal{MP}_S} = \sum_{P\::\:\omega(P,S) = 1} e\errorbasis{P}{I\otimes P} + \mu(PS) e\errorbasis{P}{Z\otimes \mu(PS)PS}$. Now consider the quantum channel where this measurement may suffer from a readout assignment error. We compose the above circuit tensor with \eqref{eq:circuit-xor-traced} to obtain the circuit weight enumerator of this channel:
    \begin{equation}\begin{split}        &\circuitoperator{\widetilde{\mathcal{MP}}_S} 
    \\&= \sum_{P\::\:\omega(P,S) = 1} (1+r)e\errorbasis{P}{I\otimes P} + (1-r)\mu(PS) e\errorbasis{P}{Z\otimes \mu(PS)PS}.
    \end{split}\end{equation}
    Note: as $r$ indicates the occurrence of a readout assignment error, one sees that neither of terms $e\errorbasis{P}{I\otimes P}$ nor $e\errorbasis{P}{Z\otimes \mu(PS)PS}$ can be associated with an assignment error, just as neither can be associated to measuring $+1$ versus $-1$.
\end{example}

In the above, we had $w_m$ and $u_\alpha$ related by the Fourier transform as implicitly we identified the error modes as elements of the integers modulo $M$. However, in some cases, this is not what we want to do. For example, consider the ($q$-ary) Pauli error channel $\widetilde{\mathcal{D}}: \mathbb{C}^{q^2} \otimes \mathfrak{H} \leadsto \mathfrak{H}$, where $\widetilde{\mathcal{D}}\left(\ketbra{m}{m} \otimes \rho\right) = P_m \rho P_m^\dagger$. Here, we want our error modes to be a pair $m = (m_1,m_2)\in \mathbb{Z}_q^2$, which indexes our Pauli operators $P_m = X^{m_1}Z^{m_2}$ (except when $q=2$ where $P_{(1,1)} = Y$).

 Based on this recognition, we take classical error basis of $\mathbb{C}^{q^2}$ to be $\{Z^\beta = Z^{\beta_1}\otimes Z^{\beta_2}\}$. Recall that 
 \begin{equation}
     P_m P_n = \zeta_q^{m_2 n_1 - m_1 n_2}P_n P_m = \zeta_q^{c(m)\cdot n} P_n P_m,
 \end{equation} 
 where $\zeta_q = e^{2\pi i/q}$ and $c(m) = (m_2, -m_1)$. Then the circuit tensor of the Pauli error channel is
\begin{align}
    &\circuittensor{\widetilde{\mathcal{D}}}{Z^\beta \otimes P_m}{P_{m'}}\nonumber\\
    &= \frac{1}{q^3} \sum_{n\in\mathbb{Z}_q^2}  \Tr((P_m^\dagger\otimes Z^{-\beta})(P^\dagger_n \otimes \ket{n}) P_{m'} (P_n\otimes \bra{n}))\nonumber\\
    &= \frac{1}{q^3} \sum_{n\in\mathbb{Z}_q^2} \zeta_q^{-n\cdot\beta} \Tr(P_m^\dagger P_n^\dagger P_{m'} P_n)\nonumber\\
    &= \frac{1}{q^3}  \sum_{n\in\mathbb{Z}_q^2} \zeta_q^{n\cdot(c(m') - \beta)} \Tr(P_m^\dagger P_{m'})\nonumber\\
    &= \begin{cases}
     1 & \text{if $m = m'$ and $\beta = c(m)$}\\
     0 &\text{otherwise.}\end{cases}
\end{align}
That is $\circuitoperator{\widetilde{\mathcal{D}}} = \sum_m e^{Z^{c(m)}\otimes P_m}_{P_m}$, and hence the trace has
\begin{equation}\label{eq:Pauli-error-enumerator}
    \widetilde{\Tr}\left[\circuitoperator{\widetilde{\mathcal{D}}}\right] = \sum_m u_{c(m)} e^{P_m}_{P_m}.
\end{equation}
Now, the relationship between our formal variables $u_\beta$ and the weight of a Pauli error $w_m$ is $u_\beta = \sum_{m\in \mathbb{Z}_q^2} \zeta_q^{\beta\cdot m} w_m$, the Fourier transform over $\mathbb{Z}_q^2$. 

In the case of a uniform Pauli error model, we take our weights as $w_{(0,0)} = w$ and $w_m = z$ when $m \not= (0,0)$. In particular, our weight function is $\mathrm{wt}(I) = (1,0)$ and $\mathrm{wt}(Z^\beta) = (0,1)$ for all $\beta \not= (0,0)$. Then our Fourier transform reduces to
\begin{equation}
    u_{(0,0)} = \sum_{m\in \mathbb{Z}_q^2} w_m = w + (q^2 - 1) z,
\end{equation}
and for $\beta \not= (0,0)$,
\begin{align}
    \nonumber u_\beta &= w + \sum_{m \not= 0} \zeta_q^{\beta\cdot m} z\\
    &= w - z + z \sum_{m \in \mathbb{Z}_q^2} \zeta_q^{\beta\cdot m} = w - z.
\end{align}
This is of course the usual MacWilliams transform for quantum codes \cite{shor1997quantum, rains1998quantum}. Hence we have shown the following result.

\begin{proposition}[Circuit enumerator of a uniform Pauli error channel]\label{proposition:pauli-error-channel}
    Let $\widetilde{\mathcal{D}}$ be the uniform ($q$-ary) Pauli error channel, where no error occurs with weight $w$ and each nontrivial Pauli error occurs with weight $z$. Then its circuit weight enumerator is
    \begin{equation}\label{eq:uniform-Pauli-circuit-enumerator}
        \widetilde{\Tr}\left[\circuitoperator{\widetilde{\mathcal{D}}}\right] = (w + (q^2 - 1)z) e\errorbasis{I}{I} + (w-z) \sum_{P\not= I} e\errorbasis{P}{P}.
    \end{equation}
\end{proposition}

\begin{example}[State preparation error]
    In Example~\ref{example:state-prep-circuit-tensor} we found that the circuit tensor for preparing a state $\ket\psi$ is given by (\ref{eq:state-prep-circuit-tensor}): $\circuittensor{\ket\psi}{}{E} = \bra\psi E \ket\psi$. A common model for state preparation error of a $n$-qubit state is to apply the uniform Pauli error channel to the prepared state. That is, the noisy state preparation circuit is $\widetilde{\mathcal{D}} \circ \ket\psi$. Composing (\ref{eq:state-prep-circuit-tensor}) and (\ref{eq:uniform-Pauli-circuit-enumerator}) gives
    \begin{equation}\begin{split}
        \widetilde{\Tr}&\left[\circuitoperator{\widetilde{\mathcal{D}}\circ \ket\psi}\right] \\&= (w + (q^2 - 1)z) e\errorbasis{}{I} + (w-z) \sum_{P\not= I} \bra\psi P \ket\psi e\errorbasis{}{E}.
    \end{split}\end{equation}
\end{example}

\begin{example}[A coherent error]
    As a very simple example of a coherent error, consider the qubit error model $\widetilde{\mathcal{S}}$ that applies $I$, with weight $w$, or $S$, $Z$, or $S^\dagger$, each with equal weight $z$. Indexing our error channels in this order, from Example~\ref{example:Clifford_ct} we have
    \begin{align}
        \circuitoperator{\widetilde{\mathcal{S}}(0)} &= e\errorbasis{I}{I} + e\errorbasis{X}{X} + e\errorbasis{Y}{Y} + e\errorbasis{Z}{Z},\\
         \circuitoperator{\widetilde{\mathcal{S}}(1)} &= e\errorbasis{I}{I} + e\errorbasis{X}{Y} - e\errorbasis{Y}{X} + e\errorbasis{Z}{Z},\\ 
        \circuitoperator{\widetilde{\mathcal{S}}(2)} &= e\errorbasis{I}{I} - e\errorbasis{X}{X} - e\errorbasis{Y}{Y} + e\errorbasis{Z}{Z},\\ 
        \circuitoperator{\widetilde{\mathcal{S}}(3)} &= e\errorbasis{I}{I} - e\errorbasis{X}{Y} + e\errorbasis{Y}{X} + e\errorbasis{Z}{Z}.
    \end{align}
    Directly applying Proposition~\ref{proposition:circuit-enumerator} we have
    \begin{align}
        \nonumber \widetilde{\Tr}\left[\circuitoperator{\widetilde{\mathcal{S}}}\right] &= (w + 3z) e\errorbasis{I}{I} + (w-z) e\errorbasis{X}{X}\\
        &\qquad + (w-z) e\errorbasis{Y}{Y} + (w + 3z) e\errorbasis{Z}{Z}.
    \end{align}
\end{example}

We will close out this section with an example of a noisy quantum teleportation. In Example~\ref{ex:teleportation} we have seen how the circuit tensor of a noise-less teleportation circuit~\eqref{eq:noiseless_teleportation} simplifies into the identity circuit tensor~\eqref{eq:ct_identity}. Now, we will add noise sources to create the circuit enumerator and utilize it to calculate the complete error model.

\begin{figure}[ht]
    \centering
    \includegraphics[width=0.95\linewidth]{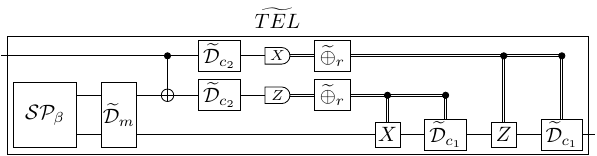}
    \caption[Noisy Teleportation Circuit]{Noisy teleportation circuit --- $\widetilde{D}_z$ represents a uniform Pauli error channel with $z$ as its formal weight vriable. $\widetilde{\oplus}_r$ is a classical bit-flip error channel with $r$ as its formal weight vriable. A classicaly controled error channel, is a noie less identity when the control bit is zero. See Fig.~\ref{fig:teleportation} for the noiseless teleportation circuit. }
    \label{fig:noisy_teleportation}
\end{figure}

\begin{example}[Noisy teleportation circuit analysis]\label{ex:noisy_teleportation}
Using the teleportation circuit given in Fig.~\ref{fig:teleportation} we will add the quantum errors in the following locations: Bell-State preparation with weight variable $m$, CNOT gate with weight variable $c_2$, and the two classically controlled Pauli operations with weight variable $c_1$. To model the quantum noise, we add a uniform Pauli error channel after each location. We will also consider classical error sources, specifically a bit flip error in both measurements, with $r$ as their weight enumeration variable. To be more precise, after the state preparation we add a 2-qubit Pauli error channel, after the CNOT we add two 1-qubit error channels one on each output leg, the noise after the classical correction is also a 1-qubit Pauli error channel but it is only applied if the classical bit is one. The noisy teleportation circuit can be seen in Fig.~\ref{fig:noisy_teleportation}.

The resulting circuit enumerator is:
\begin{equation}
    \widetilde{\Tr}\left[\circuitoperator{\widetilde{TEL}}\right] = \sum_{P\in\mathcal{P}} u_P(c_1, c_2, r, m)e\errorbasis{P}{P},
\end{equation}
where: $u_I = 1$, and 
\begin{equation}\begin{split}
u_X = u_Z= &1 -\tfrac{2}{3} c_1 -\tfrac{14}{15} c_2 -2 r -\tfrac{14}{15} m + \tfrac{1}{9} c_1^2 + \tfrac{28}{45} c_1 c_2  \\
    & + \tfrac{4}{3} c_1 r + \tfrac{28}{45} c_1 m + \tfrac{28}{15} c_2 r + \tfrac{196}{225} c_2 m + \tfrac{28}{15} r m   \\
    & -\tfrac{14}{135} c_1^2 c_2 -\tfrac{2}{9}c_1^2 r -\tfrac{14}{135} c_1^2 m -\tfrac{56}{45} c_1 c_2 r\\
    & -\tfrac{392}{675} c_1 c_2 m -\tfrac{56}{45} c_1 r m -\tfrac{392}{225} c_2 r m + \tfrac{28}{135} c_1^2 c_2 r \\
    &+ \tfrac{196}{2025} c_1^2 c_2 m + \tfrac{28}{135} c_1^2 r m  + \tfrac{784}{675} c_1 c_2 r m \\
    &-\tfrac{392}{2025} c_1^2 c_2 r m,
\end{split} \end{equation}
and
\begin{equation}\begin{split}
u_Y =& 1 -\tfrac{2}{3} c_1 -\tfrac{14}{15} c_2 -4 r -\tfrac{14}{15} m + \tfrac{1}{9} c_1^2 + \tfrac{28}{45} c_1 c_2 + \tfrac{8}{3} c_1 r\\
    &+ \tfrac{28}{45} c_1 m + \tfrac{56}{15} c_2 r + \tfrac{196}{225} c_2 m + 4 r^2 + \tfrac{56}{15} r m -\tfrac{14}{135} c_1^2 c_2 \\
    &-\tfrac{4}{9} c_1^2 r -\tfrac{14}{135} c_1^2 m -\tfrac{112}{45} c_1 c_2 r -\tfrac{392}{675} c_1 c_2 m -\tfrac{8}{3} c_1 r^2\\
    &-\tfrac{112}{45} c_1 r m -\tfrac{56}{15} c_2 r^2 -\tfrac{784}{225} c_2 r m -\tfrac{56}{15} r^2 m + \tfrac{56}{135} c_1^2 c_2 r\\
    &+ \tfrac{196}{2025} c_1^2 c_2 m + \tfrac{4}{9} c_1^2 r^2 + \tfrac{56}{135} c_1^2 r m + \tfrac{112}{45} c_1 c_2 r^2  \\
    &+\tfrac{1568}{675} c_1 c_2 r m + \tfrac{112}{45} c_1 r^2 m + \tfrac{784}{225} c_2 r^2 m -\tfrac{56}{135} c_1^2 c_2 r^2\\
    &-\tfrac{784}{2025} c_1^2 c_2 r m -\tfrac{56}{135} c_1^2 r^2 m -\tfrac{1568}{675} c_1 c_2 r^2 m\\
    &+ \tfrac{784}{2025} c_1^2 c_2 r^2 m \ .
\end{split} \end{equation}

With these circuit enumerator coefficients, we can construct an error model for the overall circuit. In this case, as we only introduce Pauli error and utilize Clifford operations, the resulting error model will be a Pauli error model. We transform the coefficient above into probabilities $p_I, p_X, p_Y, p_Z$ that satisfy
\begin{equation}
    \sum_{\text{Pauli $P$}}p_P \circuitoperator{P} = \widetilde{\Tr}\left[\circuitoperator{\widetilde{TEL}}\right].
\end{equation}
Then from (\ref{eq:ct_identity}, \ref{eq:circuit-tensor-X}-\ref{eq:circuit-tensor-Z}) we construct a set of linear equations,
\begin{equation}
    \begin{pmatrix}
        1 & 1 & 1 & 1 \\
        1 & 1 & -1 & -1 \\
        1 & -1 & 1 & -1 \\
        1 & -1 & -1 & 1 \\
    \end{pmatrix}
    \begin{pmatrix}
        p_I \\ p_X\\ p_Y \\ p_Z
    \end{pmatrix}
    =
    \begin{pmatrix}
        u_I \\ 
        u_X \\
        u_Y\\ 
        u_Z
    \end{pmatrix}.
\end{equation}
Solving these we get,
\begin{equation}\begin{split}
p_I =& 1 -\tfrac{1}{2} c_1 -\tfrac{7}{10} c_2 -2 r -\tfrac{7}{10} m + \tfrac{1}{12} c_1^2 + \tfrac{7}{15} c_1 c_2 + \tfrac{4}{3} c_1 r \\
    &+ \tfrac{7}{15} c_1 m + \tfrac{28}{15} c_2 r + \tfrac{49}{75} c_2 m + r^2 + \tfrac{28}{15} r m -\tfrac{7}{90} c_1^2 c_2 \\
    &-\tfrac{2}{9} c_1^2 r -\tfrac{7}{90} c_1^2 m -\tfrac{56}{45} c_1 c_2 r -\tfrac{98}{225} c_1 c_2 m -\tfrac{2}{3} c_1 r^2\\
    &-\tfrac{56}{45} c_1 r m -\tfrac{14}{15} c_2 r^2 -\tfrac{392}{225} c_2 r m -\tfrac{14}{15} r^2 m + \tfrac{28}{135} c_1^2 c_2 r\\
    &+ \tfrac{49}{675} c_1^2 c_2 m + \tfrac{1}{9} c_1^2 r^2 + \tfrac{28}{135} c_1^2 r m + \tfrac{28}{45} c_1 c_2 r^2 \\
    &+ \tfrac{784}{675} c_1 c_2 r m + \tfrac{28}{45} c_1 r^2 m + \tfrac{196}{225} c_2 r^2 m -\tfrac{14}{135} c_1^2 c_2 r^2 \\
    &-\tfrac{392}{2025} c_1^2 c_2 r m -\tfrac{14}{135} c_1^2 r^2 m -\tfrac{392}{675} c_1 c_2 r^2 m\\
    &+ \tfrac{196}{2025} c_1^2 c_2 r^2 m \ ,
\end{split}\end{equation}

\begin{equation}\begin{split}
p_X = p_Z =& \tfrac{1}{6} c_1 + \tfrac{7}{30} c_2 + r + \tfrac{7}{30} m -\tfrac{1}{36} c_1^2 -\tfrac{7}{45} c_1 c_2 \\
    & -\tfrac{2}{3} c_1 r -\tfrac{7}{45} c_1 m -\tfrac{14}{15} c_2 r -\tfrac{49}{225} c_2 m - r^2  \\
    & -\tfrac{14}{15} r m + \tfrac{7}{270} c_1^2 c_2 + \tfrac{1}{9} c_1^2 r + \tfrac{7}{270} c_1^2 m \\
    &+ \tfrac{28}{45} c_1 c_2 r + \tfrac{98}{675} c_1 c_2 m  + \tfrac{2}{3} c_1 r^2 + \tfrac{28}{45} c_1 r m    \\
    &+ \tfrac{14}{15} c_2 r^2 + \tfrac{196}{225} c_2 r m + \tfrac{14}{15} r^2 m -\tfrac{14}{135} c_1^2 c_2 r \\
    & -\tfrac{49}{2025} c_1^2 c_2 m  -\tfrac{1}{9} c_1^2 r^2 -\tfrac{14}{135} c_1^2 r m -\tfrac{28}{45} c_1 c_2 r^2 \\
    &-\tfrac{392}{675} c_1 c_2 r m -\tfrac{28}{45} c_1 r^2 m -\tfrac{196}{225} c_2 r^2 m  \\
    &+ \tfrac{14}{135} c_1^2 c_2 r^2 + \tfrac{196}{2025} c_1^2 c_2 r m + \tfrac{14}{135} c_1^2 r^2 m \\
    &+ \tfrac{392}{675} c_1 c_2 r^2 m -\tfrac{196}{2025} c_1^2 c_2 r^2 m\ ,
\end{split}\end{equation}
and
\begin{equation}\begin{split}
p_Y =& 
    \tfrac{1}{6} c_1 + \tfrac{7}{30} c_2 + \tfrac{7}{30} m -\tfrac{1}{36} c_1^2 -\tfrac{7}{45} c_1 c_2 -\tfrac{7}{45} c_1 m  \\
    &-\tfrac{49}{225} c_2 m + r^2 + \tfrac{7}{270} c_1^2 c_2 + \tfrac{7}{270} c_1^2 m + \tfrac{98}{675} c_1 c_2 m \\
    &-\tfrac{2}{3} c_1 r^2 -\tfrac{14}{15} c_2 r^2 -\tfrac{14}{15} r^2 m -\tfrac{49}{2025} c_1^2 c_2 m + \tfrac{1}{9} c_1^2 r^2 \\
    &+ \tfrac{28}{45} c_1 c_2 r^2 + \tfrac{28}{45} c_1 r^2 m + \tfrac{196}{225} c_2 r^2 m -\tfrac{14}{135} c_1^2 c_2 r^2 \\
    &-\tfrac{14}{135} c_1^2 r^2 m -\tfrac{392}{675} c_1 c_2 r^2 m + \tfrac{196}{2025} c_1^2 c_2 r^2 m.
\end{split}\end{equation}

\end{example}

\section{Poisson summation for stabilizer codes}\label{section:Poisson-summation}

In the previous section, we saw how various error models can be attached to circuit elements, which produce weight functions and circuit enumerators. Composing several such circuit elements together, these enumerators count, or when properly normalized compute the probability of, all error paths through the circuit. The principle of fault-tolerance is that an error mitigation strategy must also deal with errors that arise from applying error correction circuitry itself. To that end, we provide in this section a powerful computational tool for stabilizer codes akin to the Poisson Summation Formula.

Like Poisson summation, our formula arises through a duality of the convolution and pointwise product. However, in our case this duality is provided by the MacWilliams transform. Specifically, consider a Hilbert space $\mathfrak{H}$ of dimension $q$, with error basis $\mathcal{E}$. Recall from \cite[\S{VI}]{cao2023quantum}, that given a weight function $\mathrm{wt}$, an algebraic mapping $\Phi(\mathbf{u}) = (\Phi_1(\mathbf{u}), \dots, \Phi_k(\mathbf{u}))$ is a MacWilliams transform for that weight function if
\begin{equation}\label{eq:MacWilliams-transform}
    \Phi(\mathbf{u})^{\mathrm{wt}(D)} = \frac{1}{q} \sum_{E\in\mathcal{E}} \omega(D,E) \mathbf{u}^{\mathrm{wt}(E)}.
\end{equation}   

\begin{example}[Uniform Pauli errors]\label{example:Pauli-errors}
    Continuing from Proposition~\ref{proposition:pauli-error-channel}, consider the uniform Pauli error model. The weight function for this model that tracks whether a nontrivial error occurs:
    \begin{equation}
        \mathrm{wt}(E) = \begin{cases}
        (1,0) & \text{if $E = I$,}\\
        (0,1) & \text{otherwise.}
        \end{cases}
    \end{equation}
    As there, take $\mathbf{u} = (w,z)$ for a tuple of indeterminates so that
    \begin{equation}
        \mathbf{u}^{\mathrm{wt}(E)} = \begin{cases}
        w & \text{if $E = I$,}\\
        z & \text{otherwise.}\end{cases}
    \end{equation}
    For $\Phi$ to be a MacWilliams transform for this weight function, equation \eqref{eq:MacWilliams-transform} must hold. First taking $D=I$ in \eqref{eq:MacWilliams-transform} we must have
    \begin{equation}
        \Phi_0(w,z) = \frac{1}{2} \sum_{E\in\mathcal{E}} \mathbf{u}^{\mathrm{wt}(E)} = \frac{w+3z}{2}.
    \end{equation}
    On the other hand, if $D \not= I$ then on the left side of \eqref{eq:MacWilliams-transform} we have $\Phi_1(w,z)$, irrespective of $D$. On the right side \eqref{eq:MacWilliams-transform}, when $E = I$ we have $\omega(D,I) = 1$ and hence contributes $\frac{w}{2}$ to the sum. Yet, as we sum $E\in \{X,Y,Z\}$, precisely one has $\omega(D,E) = 1$, while two have $\omega(D,E) = -1$. For example, if $D = Y$ then $\omega(Y,X) = -1$, $\omega(Y,Y) = 1$, and $\omega(Y,Z) = -1$. So regardless of $D$ the sum of the three terms with $E \not= I$ contributes $-\frac{z}{2}$, and hence irrespective of $D$ we find 
    \begin{equation}
        \Phi_1(w,z) = \frac{w-z}{2}.
    \end{equation} 
    Therefore we find a unique MacWilliams transform given by 
    \begin{equation}
    \Phi(w,z) = \left(\frac{w+3z}{2}, \frac{w-z}{2}\right).
    \end{equation}
\end{example}

Given a tuple $\mathbf{wt} = (\mathrm{wt}_1, \dots, \mathrm{wt}_n)$ of weight functions, where each weight function takes values in $\mathbb{Z}_{\geq 0}^k$, we define
\begin{equation}
    \mathbf{wt}(E_1\otimes \cdots \otimes E_n) = \sum_{j=1}^n \mathrm{wt}_j(E_j).
\end{equation}

\begin{theorem}[Poisson summation for stabilizer codes]\label{theorem:Poisson-summation}
    Let $\mathfrak{C}$ be a $[[n,k,d]]_q$ stabilizer code with stabilizer group $\mathcal{S}(\mathfrak{C})$ and normalizer $\mathcal{N}(\mathfrak{C})$. Let $\mathbf{wt}_1$ and $\mathbf{wt}_2$ be scalar weight functions that have MacWilliams transforms $\Phi_1$ and $\Phi_2$ respectively. Then
    \begin{align}\label{eq:Poisson-summation1}
        \nonumber &\sum_{D\in\mathcal{N}(\mathfrak{C})}\Phi_1(\mathbf{u}_1)^{\mathbf{wt}_1(D)}\Phi_2(\mathbf{u}_2)^{\mathbf{wt}_2(D)}\\
        &\qquad = \frac{1}{q^{n-k}} \sum_{\text{\tiny$ \begin{array}{c}E_1,E_2\in\mathcal{E}^n\\E_1E_2\in\mathcal{S}(\mathfrak{C})\end{array}$}} \mathbf{u}_1^{\mathbf{wt}_1(E_1)}\mathbf{u}_2^{\mathbf{wt}_2(E_2)}
    \end{align}
    and
    \begin{align}\label{eq:Poisson-summation2}
        \nonumber &\sum_{D\in\mathcal{S}(\mathfrak{C})}\Phi_1(\mathbf{u}_1)^{\mathbf{wt}_1(D)}\Phi_2(\mathbf{u}_2)^{\mathbf{wt}_2(D)}\\ &\qquad = \frac{1}{q^{n+k}} \sum_{\text{ \tiny $\begin{array}{c}E_1,E_2\in\mathcal{E}^n\\E_1E_2\in\mathcal{N}(\mathfrak{C})\end{array}$}} \mathbf{u}_1^{\mathbf{wt}_1(E_1)}\mathbf{u}_2^{\mathbf{wt}_2(E_2)}.
    \end{align}
\end{theorem}

\begin{proof}
    Starting with the left summand, we use \eqref{eq:MacWilliams-transform} to write
    \begin{align}\label{eq:Poisson-summation-proof}
        \nonumber &\Phi_1(\mathbf{u}_1)^{\mathbf{wt}_1(D)}\Phi_2(\mathbf{u}_2)^{\mathbf{wt}_2(D)}\\
        \nonumber &\quad = \frac{1}{q^{2n}}\sum_{E_1,E_2} \omega(D,E_1)\omega(D,E_2) \mathbf{u}_1^{\mathbf{wt}_1(E_1)} \mathbf{u}_2^{\mathbf{wt}_2(E_2)}\\
        &\quad = \frac{1}{q^{2n}}\sum_{E_1,E_2} \omega(D,E_1 E_2) \mathbf{u}_1^{\mathbf{wt}_1(E_1)} \mathbf{u}_2^{\mathbf{wt}_2(E_2)}.
    \end{align}
    Now summing this equation over $D\in\mathcal{S}(\mathfrak{C})$ we note 
    \begin{equation}\label{eq:Poisson-summation-duality1}
        \sum_{D\in\mathcal{S}(\mathfrak{C})} \omega(D,E_1 E_2) = \begin{cases}
        q^{n-k} & \text{if $E_1E_2 \in \mathcal{N}(\mathfrak{C})$,}\\
        0 & \text{otherwise,}\end{cases}
    \end{equation}
    from whence \eqref{eq:Poisson-summation2} follows. Similarly, for $\mathcal{N}(\mathfrak{C})$ we have
    \begin{equation}\label{eq:Poisson-summation-duality2}
        \sum_{D\in\mathcal{N}(\mathfrak{C})} \omega(D,E_1 E_2) = \begin{cases}
         q^{n+k} & \text{if $E_1E_2 \in \mathcal{S}(\mathfrak{C})$,}\\
         0 & \text{otherwise,}\end{cases}
    \end{equation}
    and so \eqref{eq:Poisson-summation1} also follows.
\end{proof}

The crux of the above argument is \eqref{eq:Poisson-summation-proof} where we used the fact that $\omega$ is a bicharacter. Clearly, this extends to arbitrary finite products,
\begin{align}
    \nonumber &\prod_{j=1}^m \Phi_j(\mathbf{u}_j)^{\mathbf{wt}_j(D)}\\
    &\quad = \frac{1}{q^{mn}} \sum_{E_1,\dots,E_m} \omega(D,E_1\cdots E_m) \prod_{j=1}^m \mathbf{u}_j^{\mathbf{wt}_j(E_j)}.
\end{align}
Then by applying the duality relations \eqref{eq:Poisson-summation-duality1} and \eqref{eq:Poisson-summation-duality2}, we have thus proven the following extension of the theorem to arbitrary products.

\begin{corollary}[Generalized Poisson summation for stabilizer codes]\label{corollary:Poisson-summation-complete}
    Let $\mathfrak{C}$ be a $[[n,k,d]]_q$ stabilizer code with stabilizer group $\mathcal{S}(\mathfrak{C})$ and normalizer $\mathcal{N}(\mathfrak{C})$. Let $\{\mathbf{wt}_j\}_{j=1}^m$ be scalar weight functions with MacWilliams transforms $\{\Phi_j\}_{j=1}^m$ respectively. Then
    \begin{align}\label{eq:Poisson-summation-extention1}
        \nonumber &\sum_{D\in\mathcal{N}(\mathfrak{C})} \prod_{j=1}^m \Phi_j(\mathbf{u}_j)^{\mathbf{wt}_j(D)}\\ 
        &\quad = \frac{1}{q^{(m-1)n-k}} \sum_{\text{\tiny$\begin{array}{c} E_1,\dots,E_m\in\mathcal{E}^n\\E_1\cdots E_m\in\mathcal{S}(\mathfrak{C})\end{array}$}} \prod_{j=1}^m \mathbf{u}_j^{\mathbf{wt}_j(E_j)}
    \end{align}
    and
    \begin{align}\label{eq:Poisson-summation-extention2}
        \nonumber &\sum_{D\in\mathcal{S}(\mathfrak{C})} \prod_{j=1}^m \Phi_j(\mathbf{u}_j)^{\mathbf{wt}_j(D)}\\ 
        &\quad = \frac{1}{q^{(m-1)n+k}} \sum_{\text{\tiny$\begin{array}{c} E_1,\dots,E_m\in\mathcal{E}^n\\E_1\cdots E_m\in\mathcal{N}(\mathfrak{C})\end{array}$}} \prod_{j=1}^m \mathbf{u}_j^{\mathbf{wt}_j(E_j)}.
    \end{align}
\end{corollary}

We can extend Theorem~\ref{theorem:Poisson-summation} in a different direction. Given a logical operator $L$ on our code $\mathfrak{C}$ we ``twist'' the left sum by introducing $\omega(D,L^\dagger)$ into the summand. As we show in the following result, this will allow us to count error paths whose product lies in the $L$-coset of the stabilizer $\mathcal{S}(\mathfrak{C})$.

\begin{corollary}\label{col:logical_coset}
Let $L \in \mathcal{N}(\mathfrak{C})$ be any logical operator. Then
    \begin{align}
    \sum_{D \in \mathcal{N}(\mathfrak{C})} \omega(D,L^\dagger) \Phi_1(\mathbf{u}_1)^{\mathbf{wt}_1(D)}\Phi_2(\mathbf{u}_2)^{\mathbf{wt}_2(D)} \nonumber\\ = \frac{1}{q^n} \sum_{\text{\tiny$\begin{array}{c}E_1,E_2\in\mathcal{E}^n\\E_1E_2\in L\mathcal{S}(\mathfrak{C})\end{array}$}} \mathbf{u}_1^{\mathbf{wt}_1(E_1)}\mathbf{u}_2^{\mathbf{wt}_2(E_2)}.
    \end{align}
\end{corollary}

\begin{proof}
    Twisting \eqref{eq:Poisson-summation-proof} as described above we have
    \begin{align}
        \nonumber &\omega(D,L^\dagger)\Phi_1(\mathbf{u}_1)^{\mathbf{wt}_1(D)}\Phi_2(\mathbf{u}_2)^{\mathbf{wt}_2(D)}\\
        &\quad = \frac{1}{q^{2n}}\sum_{E_1,E_2} \omega(D,L^\dagger E_1 E_2) \mathbf{u}_1^{\mathbf{wt}_1(E_1)} \mathbf{u}_2^{\mathbf{wt}_2(E_2)}.
    \end{align}
    Then just as in \eqref{eq:Poisson-summation-duality2} we have
    \begin{equation}
        \sum_{D\in\mathcal{N}(\mathfrak{C})} \omega(D,L^\dagger E_1 E_2) =\begin{cases}
         q^{n+k} & \text{if $E_1E_2 \in L\mathcal{S}(\mathfrak{C})$,}\\
         0 & \text{otherwise,}\end{cases}
    \end{equation}
    from which the result follows immediately.
\end{proof}

This corollary can also be extended to arbitrary products like in Corollary~\ref{corollary:Poisson-summation-complete}. We leave the details of this to the reader.

\section{Application to Fault-tolerance}\label{section:app_ft}

The concept of quantum fault-tolerance addresses the issue that quantum operations designed to remove noise from a quantum computation are themselves noisy. Even for well designed circuits, only if the noise in these operations lies below some threshold can one guarantee that enough error correction will ultimately suppress errors and allow robust quantum computation. A critical part of this is understanding how errors accumulate circuits, with the potential of reducing the distance of the code \cite{dennis2002topological, gidney2023pair, beverland2024fault, grans2023improved}. Such ``hook'' errors can, with a single error event, produce multiple errors in the domain of a logical operation and so reduce the effective distance that a code can detect.

In many cases hook errors are rare, and so while they may reduce the effective distance of the code, they may not have a large impact on the code's error correction capability. By that same token, it is often challenging to identify when hook errors exist, how many there are, and their relative severity, particularly with Monte Carlo methods.

In this section, we show how to use circuit enumerators to analyze a syndrome extraction circuit for a stabilizer code, and provide quantitive examples for the perfect code~\cite{bennettMixedstateEntanglementQuantum1996} and the distance three and five rotated surface code~\cite{fowlerSurfaceCodesPractical2012, chao2020optimization}. Namely, we will count and characterize the error paths that generate a normalizer or a stabilizer in the syndrome extraction circuit. This method enables the quantification of all hook errors within a specific syndrome extraction circuit, along with assessing their severity, under a specified error model.

To validate the method, we have also developed a simulation tool that enumerates all possible error combinations, up to weight 3 errors, and counts all the paths that lead to the same output state as the input state (up to a global phase). When performing this simulation with different logical input states $\{\ket{0}_L, \ket{+}_L, \ket{i}_L\}$, it tallies the error paths leading to a logical $I$, along with those leading to a logical $Z$, $X$, or $Y$, dependent on the initial input state. See Fig.~\ref{fig:noisy_syndrome} for an illustration of the simulated circuit.

\begin{figure*}[ht]
    \centering
    \includegraphics{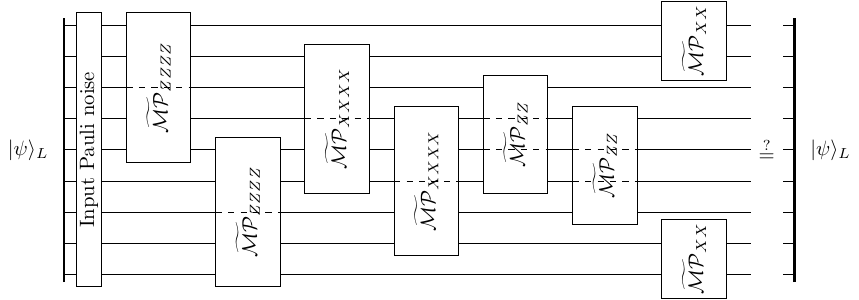}
    \caption[Noisy syndrome extraction circuit]{An illustration of the simulated circuit used to verify the enumerators' calculation result. For a distance 3 surface code, we initialize the 9 physical qubits to a logical state $\ket{\psi}_L$, then apply Pauli noise to those qubits, following are the noisy projective Pauli measurements ($\widetilde{\mathcal{M}\mathcal{P}}_{P}$) forming the syndrome extraction. We don't illustrate the measurement results, as we ignore them during the simulation. Then we count the number of different error paths that output the original input $\ket{\psi}_L$ state, which are exactly the number of error paths that generate a logical $I$ plus the number of error paths that generate a logical $Z, X,$ or $Y$ depending whether $\ket{\psi}_L$ is $\ket{0}_L$, $\ket{+}_L$, or $\ket{i}_L$.}
    \label{fig:noisy_syndrome}
\end{figure*}

The rest of this section is organized as follows: we start by developing the circuit tensor for a composition of noisy stabilizer measurements and using Theorem~\ref{theorem:Poisson-summation} (or more precisely Corollary~\ref{corollary:Poisson-summation-complete}) we show that its trace leads to the enumeration of all the error paths that generate a stabilizer.  Next, we simplify the calculation and show how when using only a simple relation between the stabilizers and normalizers we can directly count the error paths that generate a stabilizer or a normalizer. Finally, we provide more information about our validation simulation framework.

Let $\mathfrak{C}$ be a $[[n,k]]$ stabilizer code with stabilizer group $\mathcal{S} = \mathcal{S}(\mathfrak{C}) = \langle S_1 \cdots, S_{n-k}\rangle$. From Example~\ref{example:pauli-measurement}, the circuit tensor of projective measurement of stabilizer $S_j$ is:
\begin{equation}
    \circuitoperator{\mathcal{MP}_{S_j}} = \sum_{E\::\:\omega(E,S_j) = 1} (e^E_{E\otimes I} + \alpha_j(E) e^E_{(\pm ES_j)\otimes Z^{(j)}}),
\end{equation}
where $Z^{(j)} = I \otimes \cdots \otimes I \otimes I \otimes  Z \otimes I \otimes \cdots \otimes I$ is supported only in the $j$-th factor, and as in that example $\alpha_j(E) = \{\pm 1\}$ according to whether $ES_j \in \mathcal{P}^n$ or $-ES_j \in \mathcal{P}$.

A syndrome extraction circuit is a composition of stabilizer measurements. For example, the circuit tensor for two consecutive stabilizer measurements is given by:
\begin{align}
        &\circuitoperator{\mathcal{MP}_{S_1} \circ \mathcal{MP}_{S_2}}\nonumber\\ 
        &= \left[\sum_{E\in S_1^\perp} \left(e^E_{E\otimes I} + \alpha_1(E) e^E_{(\pm ES_1)\otimes Z^{(1)}}\right)\right]\nonumber\\ 
        &\qquad \cdot\left[\sum_{E\in S_2^\perp} \left(e^E_{E\otimes I} + \alpha_2(E) e^E_{(\pm ES_2)\otimes Z^{(2)}}\right)\right]\nonumber\\
        &= \sum_{E\in S_1^\perp \cap S_2^\perp} \alpha_1(E) e^E_{(\pm ES_1)\otimes Z^{(1)}} + \alpha_2(E) e^E_{(\pm ES_2)\otimes Z^{(2)}} \nonumber\\ 
        &\qquad + e^E_{E\otimes I} + \alpha_{(1,2)}(E) e^E_{(\pm ES_1S_2)\otimes Z^{(1)}Z^{(2)}}.    
\end{align}
Here, for simplicty we write $\alpha_{(1,2)}(E) = \alpha_1(E)\alpha_2(E)$ and $S_j^\perp = \{ E\in \mathcal{P}^n \::\: \omega(E,S_j) = 1\}$.

If we recursively perform the concatenation on all of the stabilizer measurements, we get the circuit tensor of the syndrome extraction circuit:
\begin{equation} \label{eq:stab_synd_ct}
\begin{split}
    &\circuitoperator{\mathcal{MP}_{S_1} \circ \cdots \circ \mathcal{MP}_{S_{n-k}}} \\
    &\quad=   \prod_{i = 1}^{n-k} \sum_{E\in S_i^\perp} (e^E_{E\otimes I} + \alpha_i(E) e^E_{(\pm ES_i)\otimes Z^{(i)}}) \\
    &\quad= \prod_{i = 1}^{n-k} \sum_{E \in \mathcal{N}(\mathfrak{C})} (e^E_{E\otimes I} + \alpha_i(E) e^E_{(\pm ES_i)\otimes Z^{(i)}}) \\
    &\quad=  \sum_{E \in \mathcal{N}(\mathfrak{C})} \prod_{i = 1}^{n-k} (e^E_{E\otimes I} + \alpha_i(E) e^E_{(\pm ES_i)\otimes Z^{(i)}}) \\
    &\quad= \sum_{E \in \mathcal{N}(\mathfrak{C})} \sum_{S \in \mathcal{S}(\mathfrak{C})} \alpha_S(E) e^E_{\pm ES\otimes Z^{\mathtt{gen}(S)}}.
\end{split}
\end{equation}

Above we write $\alpha_S(E) = \alpha_{j_1}(E) \cdots \alpha_{j_\ell}(E)$ and $Z^{\mathtt{gen}(S)} = Z^{(j_1)} \cdots Z^{(j_\ell)}$  when $S = S_{j_1}\cdot \cdots \cdot S_{j_\ell}$.

Recall that for a stabilizer code, the Shor-Laflamme $A$-enumerator \eqref{eq:Shor-Laflamme-A}, when evaluated at the projection onto the code, counts the number of stabilizers at each weight. The following result shows that the trace of the circuit tensor of its syndrome extraction circuit composed with a Pauli error channel also performs this enumeration.

\begin{proposition}\label{prop:trace_ct_is_A}
    Let $\mathfrak{C}$ have stabilizer $\mathcal{S}(\mathfrak{C}) = \langle S_1 \cdots, S_{n-k}\rangle$ and $\widetilde{\mathcal{D}}^{\otimes n}$ to be the $n$-qubit Pauli error channel. Then
    \begin{equation}
    \Tr\left[\circuitoperator{\mathcal{MP}_{S_1} \circ \cdots \circ \mathcal{MP}_{S_{n-k}}\circ \widetilde{\mathcal{D}}^{\otimes n}}\right] = A(w,z; \Pi_{\mathfrak{C}}) e_{I}.
    \end{equation}
\end{proposition}
\begin{proof}
    Expanding Proposition~\ref{proposition:pauli-error-channel} we have:
    \begin{equation}\label{eq:ct_of_n_pauli_error_ch}
        \circuitoperator{\widetilde{\mathcal{D}}^{\otimes n}} = \sum_{E\in P^n} \Phi(\mathbf{u})^{\mathbf{wt}(E)} e^E_E.
    \end{equation}
    Composing \eqref{eq:ct_of_n_pauli_error_ch} with \eqref{eq:stab_synd_ct} we get:
    \begin{equation}
    \begin{split}
       &\circuitoperator{\mathcal{MP}_{S_1} \circ \cdots \circ \mathcal{MP}_{S_{n-k}}\circ \widetilde{\mathcal{D}}^{\otimes n}} \\
       &= \sum_{E \in \mathcal{N}(\mathfrak{C})} \sum_{S \in \mathcal{S}(\mathfrak{C})} \Phi(\mathbf{u})^{\mathbf{wt}(E)} \alpha_S(E) e^E_{\pm ES\otimes Z^{\mathtt{gen}(S)}}\\
       &= \sum_{E \in \mathcal{N}(\mathfrak{C})} \Phi(\mathbf{u})^{\mathbf{wt}(E)} e^E_{E\otimes I}  + \text{ off-diagonal terms.}
    \end{split}
    \end{equation}

    Where the first part of the summation is for $S=I$. Therefore, when we take the trace we take the trace and keep only the diagonal terms we get:
    \begin{equation}\begin{split}
        &\Tr\left[\circuitoperator{\mathcal{M}_{S_1} \circ \cdots \circ \mathcal{M}_{S_{n-k}}\circ \widetilde{\mathcal{D}}^{\otimes n}}\right]\\
        &\quad=  \sum_{E \in \mathcal{N}(\mathfrak{C})} \Phi(\mathbf{u})^{\mathbf{wt}(E)} e_I.
    \end{split}\end{equation}
\end{proof}

Now we turn to the case where each operation in the circuit tensor of the syndrome extraction circuit~\eqref{eq:stab_synd_ct} is noisy. In this work, we focus on a simple error model, where each syndrome measurement $\mathcal{MP}_{S_j}$ suffers from a uniform Pauli error $\widetilde{\mathcal{D}}_j$ on its support. Let us write $\widetilde{\mathcal{MP}}_{S_j} = \widetilde{\mathcal{D}}_j\circ \mathcal{MP}_{S_j}$, and the associated variables and weight function to be $(\mathbf{wt}_j,\mathbf{u}_j)$. We can develop the following for a single syndrome measurement:
\begin{align}
    \circuitoperator{\widetilde{\mathcal{MP}}_{S_j}} &= \sum_{E\in S_j^\perp} \Phi(\mathbf{u}_j)^{\mathbf{wt}_j(E)} e^E_{E\otimes I} \nonumber\\
    &\quad\quad+ \alpha_j(E) \Phi(\mathbf{u}_j)^{\mathbf{wt}_j(ES)} e^E_{(\pm ES_j)\otimes Z^{(j)}}.
\end{align}

By concatenating all the noisy syndrome measurements together with an initial decoherence channel, for which we specify $(\mathbf{wt},\mathbf{u})$, we get:
\begin{align}\label{eq:noisy_synd_and_deco_ct}
    &\circuitoperator{\widetilde{\mathcal{M}}_{S_{1}} \circ \cdots \circ \widetilde{\mathcal{M}}_{S_{n-k}}\circ \widetilde{\mathcal{D}}^{\otimes n}} =\nonumber\\
    &\sum_{E\in\mathcal{N}(\mathfrak{C})} \Phi(\mathbf{u}_1)^{\mathbf{wt}_1(E)}\cdots \Phi(\mathbf{u}_{n-k})^{\mathbf{wt}_{n-k}(E)}  \Phi(\mathbf{u})^{\mathbf{wt}(E)} e^E_{E\otimes I}\nonumber\\
    & \qquad + \text{ off-diagonal terms.}
\end{align}

The derivation of \eqref{eq:noisy_synd_and_deco_ct} follows the same path as the derivation of \eqref{eq:stab_synd_ct} and the proof of Proposition~\ref{prop:trace_ct_is_A}.

Next, we will provide two examples~(\ref{ex:noisy_perfect_syn}) and (\ref{ex:noisy_rotated_surface_syn}) showcasing the steps needed to take starting from the stabilizers group to get the $A_{\text{path}}$ and $B_{\text{path}}$ enumerators. Rather than creating and tracing the appropriate circuit tensor, as suggested by Proposition~\ref{prop:trace_ct_is_A}, we will use Therom~\ref{theorem:Poisson-summation} and perform a simple weighted counting procedure. The noise model we will consider will be a decoherence channel before the syndrome extraction circuit, a Pauli error on the qubits that take part in a projective measurement, and Pauli idling errors for qubits outside the measurement.

To facilitate the enumeration we will use the following variables, weight functions, and MacWilliams transforms:
\begin{itemize}
    \item Initial decoherence channel will use the $z$ variable--- $\mathbf{u} = (w_z,z)$, the weight function
        \begin{equation}
        \mathtt{wt}(P) =
            \begin{cases}
                (1,0) &  P = I  \\
                (0,1) & P \in \{X,Y,Z\}
            \end{cases},
        \end{equation}
         and the MacWilliams transform of
        \begin{equation}\label{eq:pauli_macwill}
            \Phi(w_z,z) = (\tfrac{w_z+3z}{2}, \tfrac{w_z-z}{2}).
        \end{equation}
    \item Pauli errors on $r$ measured qubits for the stabilizer $S_i$ will use the $m_i$ variable--- $\mathbf{u}_{m_i} = (w_{m_i},m_i)$, the weight function
            \begin{equation} \label{eq:pauli_weight_function}
                \mathtt{wt}_{m_i}(P) =
                    \begin{cases}
                    (1,0) &  \bigotimes_{j\in \mathrm{supp}(S_i)}P_j = I^{\otimes r}  \\
                    (0,1) & \bigotimes_{j\in \mathrm{supp}(S_i)}P_j \ne I^{\otimes r}
                    \end{cases},
        \end{equation}
            and the MacWilliams transform of
            \begin{equation}\label{eq:meas_macwill}
                \Phi(w_{m_i},m_i) = (\tfrac{w_{m_i}+(4^r-1)m_i}{2^r}, \tfrac{w_{m_i}-m_i}{2^r}).
            \end{equation}
    \item Pauli errors on the idling qubit during the measurement of $S_i$ stabilizer will use the $c_i$ variable ---  $\mathbf{u}_{c_i} = (w_{c_i},c_i)$, the weight function
            \begin{align}\label{eq:pauli_weight_function_for_idle}
            \mathtt{wt}_{c_i}(P) &= (a,b) \text{ where} \nonumber \\
                &a = |\{P_j|P_j = I\}_{j\in \mathrm{off\_supp}(S_i)} | \nonumber \\ 
                &b = |\{P_j|P_j \ne I\}_{j\in \mathrm{off\_supp}(S_i)} |,
            \end{align}
            and the MacWilliams transform of
            \begin{equation}\label{eq:idling_macwill}
                \Phi(w_{c_i},c_i) = (\tfrac{w_{c_i}+3c_i}{2}, \tfrac{w_{c_i}-c_i}{2}).
            \end{equation}
\end{itemize}

We would like to emphasize that this is a worst-case error model, where a single measurement error can affect all of the measured qubits.

Just to illustrate how to apply the weight functions~\eqref{eq:pauli_weight_function} \eqref{eq:pauli_weight_function_for_idle}, consider the Pauli string $S_1 = \mathtt{XZZXI}$. This operator has support on the first four qubits, hence when we are to calculate the weight function $\mathtt{wt}_{m_1}$ on any Pauli string, we first choose the first four factors (the indexes that $S_1$ supports) and check if that is the identity $I^{\otimes 4}$. For example, for the Pauli string  $\mathtt{IXZZX}$, the supported substring is $\mathtt{IXZZ} \ne I^{\otimes 4}$ hence $\mathtt{wt}_{m_1}(\mathtt{IXZZX}) = (0,1)$. 

In a similar vein, if measuring $S_1$ only the fifth qubit is idle, and so we can compute $\mathtt{wt}_{c_1}$ on a Pauli string by checking if its fifth factor is $I$. For example, we get $\mathtt{wt}_{c_1}(\mathtt{IXZZX}) = (0,1)$ while $\mathtt{wt}_{c_1}(\mathtt{YXXYI}) = (1,0)$.

\begin{example}[Analysis of a noisy syndrome extraction for the perfect code]\label{ex:noisy_perfect_syn}
    Let us consider the generating set of $\mathcal{S} = \langle \mathtt{XZZXI}, \mathtt{IXZZX}, \mathtt{XIXZZ}, \mathtt{ZXIXZ}\rangle$. Just as above, we can calculate the appropriate weight functions for each of the elements in the set. For example, for $S_2 = \mathtt{IXZZX}$ one can easily calculate have the following:
    \begin{equation}
        \begin{tabular}{cc}
             $\mathtt{wt}_{m_2}(\mathtt{XIXZZ}) = (0,1)$ &  $\mathtt{wt}_{c_2}(\mathtt{XIXZZ}) = (0,1)$\\
             $\mathtt{wt}_{m_2}(\mathtt{IZYYZ}) = (0,1)$ &  $\mathtt{wt}_{c_2}(\mathtt{IZYYZ}) = (1,0)$
        \end{tabular}.
    \end{equation}
    Performing a simple counting operation over all normalizers and stabilizers we get the following:
    \begin{align}\label{eq:perfect_code_stab_sum}
        \sum_{E\in \mathcal{S}(\mathfrak{C})} &\mathbf{u}_{m_1}^{\mathtt{wt}_{m_1}(E)}  \mathbf{u}_{c_1}^{\mathtt{wt}_{c_1}(E)}\cdots \mathbf{u}_{m_4}^{\mathtt{wt}_{m_4}(E)}  \mathbf{u}_{c_4}^{\mathtt{wt}_{c_4}(E)} \mathbf{u}^{\mathtt{wt}(E)} = \nonumber \\ & w_c^4 w_m^4 w_z^5 + 3c^4 m^4 w_z z^4 + 12 c^3 m^4 w_c w_z z^4,
    \end{align}
    \begin{align}\label{eq:perfect_code_norm_sum}
        \sum_{E\in \mathcal{N}(\mathfrak{C})} &\mathbf{u}_{m_1}^{\mathtt{wt}_{m_1}(E)}  \mathbf{u}_{c_1}^{\mathtt{wt}_{c_1}(E)}\cdots \mathbf{u}_{m_4}^{\mathtt{wt}_{m_4}(E)} \mathbf{u}_{c_4}^{\mathtt{wt}_{c_4}(E)} \mathbf{u}^{\mathtt{wt}(E)} = \nonumber\\ & w_c^4 w_m^4 w_z^5 + 12 c^3 m^4 w_c w_z^2 z^3 + 18c^2 m^4 w_c^2 w_z^2 z^3\nonumber\\
        & + 3c^4 m^4 w_z z^4 + 12c^3 m^4 w_c w_z z^4 + 18 c^4 m^4 z^5.
    \end{align}

To get the analogues of the Shor-Laflamme enumerators $A$ and $B$, we will follow Theorem~\ref{theorem:Poisson-summation} and use the MacWilliams transforms~\eqref{eq:pauli_macwill}, \eqref{eq:meas_macwill}, \eqref{eq:idling_macwill}) on the above expressions~\eqref{eq:perfect_code_stab_sum} and \eqref{eq:perfect_code_norm_sum}. Let $A_\text{path}$ and $B_\text{path}$ be the enumerators that count the paths leading to stabilizers \eqref{eq:Poisson-summation-extention1} and normalizers \eqref{eq:Poisson-summation-extention2} respectively.

For ease of presentation, we are showing the normalized unhomogenized result up to the third order:

\begin{align}
    B_{\text{path}} &= 1 + 60m + 960mz + 12cz + 24390m^2 + 768cm \nonumber\\ & + 30z^3 + 5760mz^2 + 72cz^2 + 365760m^2z \nonumber\\ & + 11472cmz + 54c^2z + 4145340m^3 + 292608cm^2\nonumber\\ &  + 3456c^2m + 12c^3 + \cdots,
\end{align}

\begin{align}
    A_{\text{path}} &= 1 + 12m + 240mz + 12cz + 6102m^2 + 192cm \nonumber\\ & + 1440mz^2 + 91440m^2z + 2832cmz + 1036332m^3 \nonumber\\ & + 73152cm^2 + 864c^2m + \cdots.
\end{align}

We can further subtract $A_{\text{path}}$ from $B_{\text{path}}$ and get a count of all the error paths leading to a nontrivial logical error:
\begin{equation}\begin{split}
    B_{\text{path}} - A_{\text{path}} &= 48m + 720mz + 18288m^2 + 576cm + 30z^3 \\ 
    & + 4320mz^2 + 72cz^2 + 274320m^2z + 8640cmz  \\
    & + 54c^2z  + 3109008m^3 + 219456cm^2  \\
    & + 2592c^2m + 12c^3 + \cdots.
\end{split}\end{equation}
\end{example}

\begin{example}[Analysis of a noisy syndrome extraction for the distance 3 rotated surface code]\label{ex:noisy_rotated_surface_syn}
Let us consider the standard plaquettes in the distance 3 rotated surface code~\cite{fowlerSurfaceCodesPractical2012}:
\begin{align}
    \mathcal{S} &= \langle \mathtt{ZZIZZIIII},\mathtt{IIIIZZIZZ},\mathtt{IXXIXXIII},\mathtt{IIIXXIXXI},\nonumber\\ 
    &\qquad \mathtt{IIZIIZIII},\mathtt{IIIZIIZII},\mathtt{XXIIIIIII},\mathtt{IIIIIIIXX}\rangle.
\end{align}
Some of the plaquettes have 4 qubits while others have only 2, thus we want to distinguish between them in our enumeration. We will do that by using two different variables $m_4$ and $m_2$. We will use the same counting variables and weight functions for the initial errors and the idling errors as in Example~\ref{ex:noisy_perfect_syn}, while for the measurement errors, we have $m_4$ and $m_2$ as counting variables and the respective $\mathtt{wt}_{m_4}$ and $\mathtt{wt}_{m_2}$ weight functions. 

Table~\ref{tab:rotated_surface_code_weight} contains the (unhomogenized) monomials for 8 out of the 256 stabilizers. When summing over all the products of the stabilizers we get the following unhomogenized polynomial:
\begin{align}
    &1 + 4m_4^2 m_2 c^{11} z^2 + 8 m_4^4 m_2^2 c^{20} z^4 + 8 m_4^4 m_2^2 c^{21} z^4 \nonumber \\
    & \ + 4 m_4^3 m_2^2 c^{22} z^4 + 2 m_4^4 m_2^2 c^{22} z^4 + 32 m_4^4 m_2^3 c^{31} z^6 \nonumber\\
    &\ + 32 m_4^4 m_2^3 c^{32} z^6 + 8 m_4^4 m_2^4 c^{32} z^6 + 4 m_4^4 m_2^3 c^{33} z^6 \nonumber\\
    &\ + 16 m_4^4 m_2^4 c^{33} z^6 + 8 m_4^4 m_2^4 c^{34} z^6 + 56 m_4^4 m_2^4 c^{42} z^8 \nonumber\\
    &\ + 56 m_4^4 m_2^4 c^{43} z^8 + 17 m_4^4 m_2^4 c^{44} z^8.
\end{align}

Following Therom~\ref{theorem:Poisson-summation} we can calculate the $A_{\text{path}}$ and $B_{\text{path}}$ enumerators. For ease of presentation, we are showing the normalized unhomogenized result up to the third order:

\begin{align}\label{eq:surface_b_enumerator}
    B_{\text{path}} &= 1 + 16m + 438c^2 + 188cz + 1952 c m + 4 z^2 \nonumber \\
    & + 368 m z  + 3228 m^2  + 5432 c^3 + 3358 c^2 z + 92600 c^2 m \nonumber \\
    & + 432 c z^2  + 36160 c z m + 395744 c m^2  + 2832 m z^2 \nonumber \\
    & + 74600 m^2 z + 403280 m^3 + 24 z^3 + \cdots,
\end{align}

\begin{align}\label{eq:surface_a_enumerator}
    A_{\text{path}} &= 1 + 16m + 438c^2 + 188cz + 1320 c m + 4 z^2   \nonumber \\
    & + 256 m z + 1516 m^2  + 1824 c^3 + 1316 c^2 z + 44224 c^2 m  \nonumber \\ 
    & + 88 c z^2 + 17992 c z m + 150744 c m^2 + 1264 m z^2  \nonumber \\ 
    & + 28880 m^2 z + 114192 m^3 + \cdots.
\end{align}

We can further subtract $A_{\text{path}}$ from $B_{\text{path}}$ and get a count of all the error paths leading to a nontrivial logical error:
\begin{align}
    B_{\text{path}} &- A_{\text{path}} = 632 c m + 112 m z  + 1712 m^2 + 3608 c^3  \nonumber \\
    & + 2042 c^2 z + 48376 c^2 m + 344 c z^2 + 18168 c z m  \nonumber \\
    & + 245000 c m^2 + 1568 m z^2 + 45720 m^2 z + 289088 m^3 \nonumber \\
    & + 24 z^3 + \cdots.
\end{align}

When discarding any idling errors in \eqref{eq:surface_a_enumerator} we have 144,336 error paths with 3 errors that lead to a logical $I$. Using the observation in Corollary~\ref{col:logical_coset} we can calculate the error paths enumerators that lead to any specific logical error in the syndrome extraction circuit. For the logical $X$ or $Z$ errors, the number of paths with 3 errors is 120,260 while for the logical $Y$ the number is 95,880.
\end{example}

\begin{table*}[htpb]
    \centering
    \begin{tabular}{c|ccccccccc|c}
Stabilizer & $\mathbf{u}_1^{\mathrm{wt}_1}$ & $\mathbf{u}_2^{\mathrm{wt}_2}$ & $\mathbf{u}_3^{\mathrm{wt}_3}$ & $\mathbf{u}_4^{\mathrm{wt}_4}$ & $\mathbf{u}_5^{\mathrm{wt}_5}$ & $\mathbf{u}_6^{\mathrm{wt}_6}$ & $\mathbf{u}_7^{\mathrm{wt}_7}$ & $\mathbf{u}_8^{\mathrm{wt}_8}$ & $\mathbf{u}^{\mathrm{wt}}$ & product\\
\hline
$\mathtt{ZZIZZIIII}$ & $m_4     $ & $ m_4 c^3 $ & $ m_4 c^2 $ & $ m_4 c^2 $ & $ c^4     $ & $ m_2 c^3 $ & $ m_2 c^2 $ & $ c^4     $ & $ z^4 $ & $ m_4^4 m_2^2 c^{20} z^4$ \\
$\mathtt{IIIIZZIZZ} $ & $ m_4 c^3 $ & $ m_4     $ & $ m_4 c^2 $ & $ m_4 c^2 $ & $ m_2 c^3 $ & $ c^4     $ & $ c^4     $ & $ m_2 c^2 $ & $ z^4 $ & $ m_4^4 m_2^2 c^{20} z^4$ \\
$\mathtt{IXXIXXIII} $ & $ m_4 c^2 $ & $ m_4 c^2 $ & $ m_4     $ & $ m_4 c^3 $ & $ m_2 c^2 $ & $ c^4     $ & $ m_2 c^3 $ & $ c^4     $ & $ z^4 $ & $ m_4^4 m_2^2 c^{20} z^4$ \\
$\mathtt{IIIXXIXXI} $ & $ m_4 c^2 $ & $ m_4 c^2 $ & $ m_4 c^3 $ & $ m_4     $ & $ c^4     $ & $ m_2 c^2 $ & $ c^4     $ & $ m_2 c^3 $ & $ z^4 $ & $ m_4^4 m_2^2 c^{20} z^4$ \\
$\mathtt{IIZIIZIII} $ & $ c^2     $ & $ m_4 c   $ & $ m_4     $ & $ c^2     $ & $ m_2     $ & $ c^2     $ & $ c^2     $ & $ c^2     $ & $ z^2 $ & $ m_4^2 m_2 c^{11} z^2$ \\
$\mathtt{IIIZIIZII} $ & $ m_4 c   $ & $ c^2     $ & $ c^2     $ & $ m_4     $ & $ c^2     $ & $ m_2     $ & $ c^2     $ & $ c^2     $ & $ z^2 $ & $ m_4^2 m_2 c^{11} z^2$ \\
$\mathtt{XXIIIIIII} $ & $ m_4     $ & $ c^2     $ & $ m_4 c   $ & $ c^2     $ & $ c^2     $ & $ c^2     $ & $ m_2     $ & $ c^2     $ & $ z^2 $ & $ m_4^2 m_2 c^{11} z^2$ \\
$\mathtt{IIIIIIIXX} $ & $ c^2     $ & $ m_4     $ & $ c^2     $ & $ m_4 c   $ & $ c^2     $ & $ c^2     $ & $ c^2     $ & $ m_2     $ & $ z^2 $ & $ m_4^2 m_2 c^{11} z^2$
    \end{tabular}
    \caption{The (unhomogenized) variables with their appropriate weight for the distance 3 rotated surface code. We show only 8 out of the 256  stabilizers }
    \label{tab:rotated_surface_code_weight}
\end{table*}

Next, we describe the simulation we used to validate the results of Example~\ref{ex:noisy_rotated_surface_syn}. 

The noise model we considered in the simulation has two types of errors. The first is a Pauli error affecting one or more of the data qubits even before starting the syndrome extraction circuit. The second error is a post-measurement Pauli error, that is applied to one or more of the plaquettes. For example, in the 4-qubit plaquette, the possible errors are any four Pauli except $IIII$. See Fig.~\ref{fig:noisy_syndrome} for the simulated circuit and more information about the counting.

To generate the logical zero state we initialized all of the data qubits to $\ket{0}$ and then performed a noiseless syndrome extraction. Similarly, to generate the $\ket{+}_L$ state we initialized all of the data qubits to $\ket{+}$ before performing the noiseless syndrome extraction. Before running the syndrome extraction circuit with noise in different locations, we initialized the data qubits to one of the saved logical states.

In total, we considered 104,183,380 possible combinations of errors. The simulation time for each different initial state was about 1000 core hours. The circuit enumerator calculations took only a few seconds using Sage~\cite{sagemath}.

When initializing the data qubits to $\ket{0}_L$, we got 264,596 possible error paths, with 3 errors, that generate a logical $I$ or a logical $Z$. The same number of error paths was calculated when we initialized the data qubits to $\ket{+}_L$. These numbers align perfectly with those calculated in Example~\ref{ex:noisy_rotated_surface_syn}, as $144336 + 120260 = 264596$.

We conclude this section with the $A_{\text{path}}$ and $B_{\text{path}}$ enumerators for the distance five rotated surface code syndrome extraction circuit. In Example~\ref{ex:noisy_rotated_surface_syn} we have shown how to efficiently calculate these enumerators for the distance 3 rotated surface code syndrome extraction circuit. 

We will use the standard plaquettes of the distance five rotated surface code, as shown in Fig.~\ref{fig:d5_sc_plq}. The generator group we use can be seen in Appendix~\ref{app:d5_rsc_gen}. The enumerators' calculation for the distance five code follows the same path as presented in Example~\ref{ex:noisy_rotated_surface_syn}, resulting in the error paths enumerators below. We present the results up to degree 5.

\newcommand{\halfscaling}{0.6}
\newcommand{\nodesize}{1.5}

\begin{figure}[ht]
    \centering
    \includegraphics{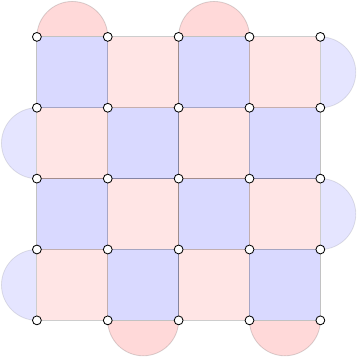}
    \caption[Distance five rotated surface code plaquettes]{Standard plaquettes of the distance five rotated surface code --- The vertices represent the physical qubits, while the purple and pink patches represent $Z$ and $X$ stabilizers respectively. The complete generator group can be found in Appendix~\ref{app:d5_rsc_gen}.}
    \label{fig:d5_sc_plq}
\end{figure}

\begin{equation}\begin{split}
    B_{\text{path}} &= 1 + 40 m + 8 z^2 + 704 m z + 4892 m^2 + 3656 m z^2 \\
    &+ 106568 m^2 z + 606632 m^3 + 72 z^4 + 16960 m z^3\\
    &+ 1156208 m^2 z^2 + 19015984 m^3 z + 94658202 m^4  \\
    & + 160 z^5 + 73040 m z^4 + 8544672 m^2 z^3  \\
    & + 292544120 m^3 z^2 + 3723068248 m^4 z \\
    &+ 16168935704 m^5 + \cdots
\end{split} \end{equation}

\begin{equation} \begin{split}
    A_{\text{path}} &= 1 + 40 m + 8 z^2 + 704 m z + 4892 m^2 + 3656 m z^2  \\ 
    &+ 103440 m^2 z + 548712 m^3 + 72 z^4 + 15424 m z^3 \\
    &+ 1046000 m^2 z^2 + 15997312 m^3 z + 71438618 m^4\\
    &+ 52816 m z^4 + 6800352 m^2 z^3 + 222326424 m^3 z^2 \\
    &+ 2569524432 m^4 z + 9919808920 m^5 + \cdots
\end{split} \end{equation}

\begin{equation} \begin{split}
    B_{\text{path}} - A_{\text{path}} &=  57920 m^3  + 3128 m^2 z + 110208 m^2 z^2 \\
    &+ 3018672 m^3 z + 1536 m z^3 +  1744320 m^2 z^3 \\
    &+  20224 m z^4 + 6249126784 m^5  \\
    &+ 1153543816 m^4 z + 23219584 m^4 \\
    &+ 70217696 m^3 z^2  +  160 z^5 + \cdots
\end{split} \end{equation}

Executing the enumerators' calculation required approximately 30 hours using our Sage code running on a single laptop CPU core. Given our current approach validation approach, conducting such a simulation would be unfeasible due to the necessity of considering 6,395,354,893,463,716 potential error paths.








\section{Conclusions}\label{section:concliu}

In this paper, we presented a novel approach to analyze circuits and error models in the context of quantum error correction. By introducing the concept of circuit enumerators, we have extended the framework of tensor enumerators to address the specific challenges posed by fault-tolerant quantum circuits. We developed all the necessary machinery to analyze any quantum circuit.  Our methodology offers a promising alternative to Monte Carlo techniques, providing an explicit means of analyzing circuits without sacrificing accuracy.

With the development of an analogue of the Poisson summation formula tailored for stabilizer codes, we have demonstrated the efficacy of our approach in rapidly enumerating error paths within syndrome extraction circuits. This advancement is particularly significant as it enables the exact computation of error probabilities, even for larger circuits and rare error events, where traditional computational methods become intractable.

The explicit counting of error paths in a distance five surface code serves as a great example of the utility and effectiveness of our approach. By showcasing scenarios previously deemed infeasible via simulation, we have highlighted the practical relevance of circuit enumerators in the realm of quantum error correction.

Looking forward, the development and usage of circuit enumerators hold promise for advancing the field of quantum computation. Future research may focus on using the methodology to tackle complex quantum circuits other than syndrome extractions. Furthermore, in our example, we ignored the syndrome measurement result, while there is value in taking that into account as well.

\bibliography{main}

\begin{thebibliography}{10}

\bibitem{shor1996fault}
Peter~W Shor.
\newblock Fault-tolerant quantum computation.
\newblock In {\em Proceedings of 37th conference on foundations of computer
  science}, pages 56--65. IEEE, 1996.

\bibitem{gottesman1998theory}
Daniel Gottesman.
\newblock Theory of fault-tolerant quantum computation.
\newblock {\em Physical Review A}, 57(1):127, 1998.

\bibitem{preskill1998fault}
John Preskill.
\newblock Fault-tolerant quantum computation.
\newblock In {\em Introduction to quantum computation and information}, pages
  213--269. World Scientific, 1998.

\bibitem{knill2004fault1}
Emanuel Knill.
\newblock Fault-tolerant postselected quantum computation: Schemes.
\newblock {\em arXiv preprint quant-ph/0402171}, 2004.

\bibitem{knill2004fault2}
Emanuel Knill.
\newblock Fault-tolerant postselected quantum computation: Threshold analysis.
\newblock {\em arXiv preprint quant-ph/0404104}, 2004.

\bibitem{wang2009threshold}
David~S Wang, Austin~G Fowler, Ashley~M Stephens, and Lloyd~CL Hollenberg.
\newblock Threshold error rates for the toric and surface codes.
\newblock {\em arXiv preprint arXiv:0905.0531}, 2009.

\bibitem{dennis2002topological}
Eric Dennis, Alexei Kitaev, Andrew Landahl, and John Preskill.
\newblock Topological quantum memory.
\newblock {\em Journal of Mathematical Physics}, 43(9):4452--4505, 2002.

\bibitem{beverland2024fault}
Michael~E Beverland, Shilin Huang, and Vadym Kliuchnikov.
\newblock Fault tolerance of stabilizer channels.
\newblock {\em arXiv preprint arXiv:2401.12017}, 2024.

\bibitem{cao2023quantum}
ChunJun Cao and Brad Lackey.
\newblock Quantum weight enumerators and tensor networks.
\newblock {\em IEEE Transactions on Information Theory}, December 2023.

\bibitem{knill1996group}
Emanuel Knill.
\newblock Group representations, error bases and quantum codes.
\newblock {\em arXiv preprint quant-ph/9608049}, 1996.

\bibitem{knill1996non}
Emanuel Knill.
\newblock Non-binary unitary error bases and quantum codes.
\newblock {\em arXiv preprint quant-ph/9608048}, 1996.

\bibitem{klappenecker2002beyond}
AA~Klappenecker and M~Rotteler.
\newblock Beyond stabilizer codes {I}. {N}ice error bases.
\newblock {\em IEEE Transactions on Information Theory}, 48(8):2392--2395,
  2002.

\bibitem{Nielsen_Chuang_2010}
Michael~A. Nielsen and Isaac~L. Chuang.
\newblock {\em Quantum Computation and Quantum Information: 10th Anniversary
  Edition}.
\newblock Cambridge University Press, 2010.

\bibitem{shor1997quantum}
Peter Shor and Raymond Laflamme.
\newblock Quantum analog of the {M}ac{W}illiams identities for classical coding
  theory.
\newblock {\em Physical review letters}, 78(8):1600, 1997.

\bibitem{rains1998quantum}
Eric~M Rains.
\newblock Quantum weight enumerators.
\newblock {\em IEEE Transactions on Information Theory}, 44(4):1388--1394,
  1998.

\bibitem{hu2019complete}
Chuangqiang Hu, Shudi Yang, and Stephen S-T Yau.
\newblock Complete weight distributions and {M}ac{W}illiams identities for
  asymmetric quantum codes.
\newblock {\em IEEE Access}, 7:68404--68414, 2019.

\bibitem{hu2020weight}
Chuangqiang Hu, Shudi Yang, and Stephen S-T Yau.
\newblock Weight enumerators for nonbinary asymmetric quantum codes and their
  applications.
\newblock {\em Advances in Applied Mathematics}, 121:102085, 2020.

\bibitem{delfosse2023spacetime}
Nicolas Delfosse and Adam Paetznick.
\newblock Spacetime codes of clifford circuits.
\newblock {\em arXiv preprint arXiv:2304.05943}, 2023.

\bibitem{cirac2021matrix}
J~Ignacio Cirac, David Perez-Garcia, Norbert Schuch, and Frank Verstraete.
\newblock Matrix product states and projected entangled pair states: Concepts,
  symmetries, theorems.
\newblock {\em Reviews of Modern Physics}, 93(4):045003, 2021.

\bibitem{zyczkowski2004duality}
Karol {\.Z}yczkowski and Ingemar Bengtsson.
\newblock On duality between quantum maps and quantum states.
\newblock {\em Open systems \& information dynamics}, 11(1):3--42, 2004.

\bibitem{jiang2013channel}
Min Jiang, Shunlong Luo, and Shuangshuang Fu.
\newblock Channel-state duality.
\newblock {\em Physical Review A}, 87(2):022310, 2013.

\bibitem{Gosset2014}
David Gosset, Vadym Kliuchnikov, Michele Mosca, and Vincent Russo.
\newblock An algorithm for the t-count.
\newblock {\em Quantum Info. Comput.}, 14(15--16):1261--1276, 2014.

\bibitem{gottesman1998heisenberg}
Daniel Gottesman.
\newblock The {{Heisenberg Representation}} of {{Quantum Computers}}.
\newblock {\em arXiv preprint quant-ph/9807006}, 1998.

\bibitem{choi1975completely}
Man-Duen Choi.
\newblock Completely positive linear maps on complex matrices.
\newblock {\em Linear algebra and its applications}, 10(3):285--290, 1975.

\bibitem{litinski2019game}
Daniel Litinski.
\newblock A {{Game}} of {{Surface Codes}}: {{Large-Scale Quantum Computing}}
  with {{Lattice Surgery}}.
\newblock {\em Quantum}, 3:128, 2019.

\bibitem{paetznick2013repeat}
Adam Paetznick and Krysta~M Svore.
\newblock Repeat-until-success: Non-deterministic decomposition of single-qubit
  unitaries.
\newblock {\em arXiv preprint arXiv:1311.1074}, 2013.

\bibitem{wiebe2014quantum}
Nathan Wiebe and Martin Roetteler.
\newblock Quantum arithmetic and numerical analysis using repeat-until-success
  circuits.
\newblock {\em arXiv preprint arXiv:1406.2040}, 2014.

\bibitem{gidney2023pair}
Craig Gidney.
\newblock A pair measurement surface code on pentagons.
\newblock {\em Quantum}, 7:1156, 2023.

\bibitem{grans2023improved}
Linnea Grans-Samuelsson, Ryan~V Mishmash, David Aasen, Christina Knapp, Bela
  Bauer, Brad Lackey, Marcus~P da~Silva, and Parsa Bonderson.
\newblock Improved pairwise measurement-based surface code.
\newblock {\em arXiv preprint arXiv:2310.12981}, 2023.

\bibitem{bennettMixedstateEntanglementQuantum1996}
Charles~H. Bennett, David~P. DiVincenzo, John~A. Smolin, and William~K.
  Wootters.
\newblock Mixed-state entanglement and quantum error correction.
\newblock {\em Physical Review A}, 54(5):3824--3851, November 1996.

\bibitem{fowlerSurfaceCodesPractical2012}
Austin~G. Fowler, Matteo Mariantoni, John~M. Martinis, and Andrew~N. Cleland.
\newblock Surface codes: {{Towards}} practical large-scale quantum computation.
\newblock {\em Physical Review A}, 86(3):032324, September 2012.

\bibitem{chao2020optimization}
Rui Chao, Michael~E Beverland, Nicolas Delfosse, and Jeongwan Haah.
\newblock Optimization of the surface code design for majorana-based qubits.
\newblock {\em Quantum}, 4:352, 2020.

\bibitem{sagemath}
{The Sage Developers}.
\newblock {\em {S}ageMath, the {S}age {M}athematics {S}oftware {S}ystem
  ({V}ersion 9.5)}, 2024.
\newblock {\tt https://www.sagemath.org}.

\end{thebibliography}

\appendices

\section{Stablizer Generator Group for the distance five rotated surface code}\label{app:d5_rsc_gen}
For the calculation of enumerators for the distance five rotated surface code we have used the following stabilizer generator group:

\begin{equation}\begin{split}
\mathcal{S} = \langle &\mathtt{ZZIIIZZIIIIIIIIIIIIIIIIII},\\ &\mathtt{ IIZZIIIZZIIIIIIIIIIIIIIII},\\ &\mathtt{ IIIIIIZZIIIZZIIIIIIIIIIII},\\ &\mathtt{ IIIIIIIIZZIIIZZIIIIIIIIII},\\ &\mathtt{ IIIIIIIIIIZZIIIZZIIIIIIII},\\ &\mathtt{ IIIIIIIIIIIIZZIIIZZIIIIII},\\ &\mathtt{ IIIIIIIIIIIIIIIIZZIIIZZII},\\ &\mathtt{ IIIIIIIIIIIIIIIIIIZZIIIZZ},\\ &\mathtt{ IXXIIIXXIIIIIIIIIIIIIIIII},\\ &\mathtt{ IIIXXIIIXXIIIIIIIIIIIIIII},\\ &\mathtt{ IIIIIXXIIIXXIIIIIIIIIIIII},\\ &\mathtt{ IIIIIIIXXIIIXXIIIIIIIIIII},\\ &\mathtt{ IIIIIIIIIIIXXIIIXXIIIIIII},\\ &\mathtt{ IIIIIIIIIIIIIXXIIIXXIIIII},\\ &\mathtt{ IIIIIIIIIIIIIIIXXIIIXXIII},\\ &\mathtt{ IIIIIIIIIIIIIIIIIXXIIIXXI},\\ &\mathtt{ IIIIZIIIIZIIIIIIIIIIIIIII},\\ &\mathtt{ IIIIIZIIIIZIIIIIIIIIIIIII},\\ &\mathtt{ IIIIIIIIIIIIIIZIIIIZIIIII},\\ &\mathtt{ IIIIIIIIIIIIIIIZIIIIZIIII},\\ &\mathtt{ XXIIIIIIIIIIIIIIIIIIIIIII},\\ &\mathtt{ IIIIIIIIIIIIIIIIIIIIIXXII},\\ &\mathtt{ IIXXIIIIIIIIIIIIIIIIIIIII},\\ &\mathtt{ IIIIIIIIIIIIIIIIIIIIIIIXX}\rangle
\end{split}\end{equation}

\section{Examples of circuit tensor for boolean operations}\label{app:calssical_ct_ex}
In this appendix, we provide more examples of circuit tensors for some primitive boolean operations that commonly appear in quantum fault-tolerance circuits.

In Example~\ref{exp:xor} we constructed the circuit tensor for the xor operation. Here we provide the construction for the and, or, and mux functions.

\begin{example}[And circuit tensor]\label{exp:and}
    The operator of \textbf{and} is: 
    \begin{equation}
        A_{and} = \sum_{x_0,x_1} \ket{x_0 \land x_1}\bra{x_0,x_1} \in L((\mathbb{C}^2)^{\otimes 2},\mathbb{C}^2).
    \end{equation}
    For it, we get the circuit tensor:
\begin{align}
    &\circuittensor{\mathtt{and}}{Z^{\alpha_0}\otimes Z^{\alpha_1}}{Z^\beta} = \tfrac{1}{4} \sum_{x_0,x_1} (-1)^{-\alpha_0 x_0 - \alpha_1 x_1 + \beta (x_0 \land x_1)} \nonumber\\
    &\qquad\qquad= \begin{cases}
    1 & \text{ if $\alpha_0 = \alpha_1 = \beta=0$,}\\
    (-1)^{\alpha_0 \alpha_1}\frac{1}{2} & \text{ if $\beta=1$,}\\
    0 & \text{otherwise.}\end{cases}
\end{align}
 Or using a tensor basis 
 \begin{equation}\label{eq:and_ct}
     \circuitoperator{\mathtt{and}} = e\errorbasis{I\otimes I}{I} + \frac{1}{2}(e\errorbasis{I\otimes I}{Z} + e\errorbasis{I\otimes Z}{Z} + e\errorbasis{Z\otimes I}{Z} - e\errorbasis{Z\otimes Z}{Z}).
 \end{equation}
\end{example}

\begin{example}[Or circuit tensor]\label{exp:or}
    The operator of \textbf{or} is:
    \begin{equation}
    A_{or} = \sum_{x_0,x_1} \ket{x_0 \lor  x_1}\bra{x_0,x_1} \in L((\mathbb{C}^2)^{\otimes 2},\mathbb{C}^2).
    \end{equation}
    For it, we get the circuit tensor:
\begin{align}
    &\circuittensor{\mathtt{or}}{Z^{\alpha_0}\otimes Z^{\alpha_1}}{Z^\beta} = \tfrac{1}{4} \sum_{x_0,x_1} (-1)^{-\alpha_0 x_0 - \alpha_1 x_1 + \beta (x_0 \lor x_1)} \nonumber\\
    &\qquad = \begin{cases}
     1 & \text{ if $\alpha_0 = \alpha_1 = \beta=0$,}\\
     (-1)^{1 + \alpha_0\alpha_1}\frac{1}{2} & \text{ if $\beta=1$,}\\
     0 & \text{otherwise.}
     \end{cases}
\end{align}
 Or using a tensor basis 
 \begin{equation}
 \circuitoperator{\mathtt{or}} = e\errorbasis{I\otimes I}{I} + \frac{1}{2}( - e\errorbasis{I\otimes I}{Z} + e\errorbasis{I\otimes Z}{Z} + e\errorbasis{Z\otimes I}{Z} + e\errorbasis{Z\otimes Z}{Z}).
 \end{equation}
\end{example}

\begin{example}[Mux circuit tensor]\label{exp:mux}
    The operator of \textbf{mux} is: 
    \begin{equation}    
    A_{mux} = \sum_{s, x_1,x_2} \ket{mux(s, x_1, x_2)}\bra{s,x_1,x_2} \in L((\mathbb{C}^2)^{\otimes 3},\mathbb{C}^2).
    \end{equation}
    For it, we get the circuit tensor:
\begin{align}
    &\circuittensor{\mathtt{mux}}{Z^{\alpha_0}\otimes Z^{\alpha_1} \otimes Z^{\alpha_2}}{Z^\beta} \nonumber\\
    &\quad =\tfrac{1}{8} \sum_{s, x_1, x_2} (-1)^{\alpha_0 s + \alpha_1 x_1 + \alpha_2 x_2 + \beta\ mux(s, x_1,  x_2)} \nonumber\\
    &\quad = \begin{cases}
     1 & \text{ if $\alpha_0 = \alpha_1 = \alpha_2 = \beta=0$,}\\
     (-1)^{\alpha_0\alpha_2}\frac{1}{2} & \text{ if $\beta=\alpha_1\oplus\alpha_2 = 1$,}\\
     0 & \text{otherwise.}\end{cases}
\end{align}
 Or using a tensor basis 
 \begin{align}
 &\circuitoperator{\mathtt{mux}} = e\errorbasis{I\otimes I\otimes I}{I}\nonumber\\
 &\qquad + \frac{1}{2}(e\errorbasis{I\otimes I \otimes Z}{Z} + e\errorbasis{I\otimes Z \otimes I}{Z} - e\errorbasis{Z\otimes I \otimes Z}{Z} + e\errorbasis{Z\otimes Z \otimes I}{Z}).
 \end{align}
\end{example}

\section{Relation between process matrix and circuit tensor}\label{app:pm_ct}
In this work, we have constructed circuit tensors from the operator representation, or Choi-Kraus form, of a quantum channel. However, another popular representation of a quantum channel arises from its so-called process matrix (usually computed relative to the Pauli basis, but any error basis will suffice). In this appendix, we show that the process matrix of a quantum channel is related to the channel's circuit tensor precisely by the quantum MacWilliams identity for tensor enumerators, \cite{cao2023quantum}.

\begin{definition}[Process matrix of a quantum channel]
    Let $\mathcal{E}$ be an error basis on a Hilbert space $\mathfrak{H}$ and $\mathcal{M}:\mathfrak{H} \leadsto \mathfrak{H}$ a quantum channel. The \emph{process matrix} of $\mathcal{M}$ in the error basis $\mathcal{E}$ is the matrix $\chi\errorbasis{E}{E'}$ defined implicitly by $\mathcal{M}(\rho) = \sum_{E,E'} \chi\errorbasis{E}{E'} (E')^\dagger \rho E$.
\end{definition}

Unlike in the case of Kraus operators, the process matrix associated with a quantum channel is unique and completely characterizes the channel. To see this, we extend $\mathcal{M}$ linearly to all linear operators $L(\mathfrak{H})$ and then we could define the process matrix directly as $\chi\errorbasis{E}{E'} = \frac{1}{\dim(\mathcal{H})} \Tr(E^\dagger \mathcal{M}(E'))$.

Like the circuit tensor, we write the process matrix in an index-free form $\chi = \frac{1}{\dim(\mathcal{H})} \sum_{E,E'} \Tr(E^\dagger \mathcal{M}(E')) e\errorbasis{E}{E'}$.

Note that the matrix is a Hermitian matrix: by definition
\begin{align}
    \nonumber \mathcal{M}(\rho) &= (\mathcal{M}(\rho))^\dagger = \sum_{E,E'} \overline{\chi\errorbasis{E}{E'}} E^\dagger \rho E'\\
    &= \sum_{E,E'} \overline{\chi\errorbasis{E'}{E}} (E')^\dagger \rho E.
\end{align}
So by uniqueness of the process matrix, $\overline{\chi\errorbasis{E'}{E}} = \chi\errorbasis{E}{E'}$.

It is straightforward to compute the process matrix from a channel's Kraus operators. Namely, if $\mathcal{M}(\rho) = \sum_j A_j \rho A_j^\dagger$ then we expand each Kraus operator as $A_j = \sum_E \lambda_{j,E} E^\dagger$. Then substituting this expression gives
\begin{equation}
    \mathcal{M}(\rho) = \sum_j \left(\sum_{E,E'} \lambda_{j,E'}\overline{\lambda_{j,E}}\right) (E')^\dagger \rho E.
\end{equation}
And hence we see $\chi^E_{E'} = \sum_j \lambda_{j,E'}\overline{\lambda_{j,E}}$. Conversely finding Kraus operators from the process matrix of a channel just is a matter of diagonalization.

Following \cite{cao2023quantum}, and akin to what appears in \S{\ref{section:Poisson-summation}}, we define the quantum MacWilliams transform of an error basis $\mathcal{E}$ to be the linear operator on matrices indexed by elements of $\mathcal{E}$ given by
\begin{equation}
    \Psi(e\errorbasis{E}{E'}) = \frac{1}{\dim(\mathfrak{H})^2} \sum_{F,F'\in\mathcal{E}} \Tr(F^\dagger E F' (E')^\dagger) e\errorbasis{F}{F'}.
\end{equation}
Then our version of the MacWilliams identity for quantum channels is as follows.

\begin{theorem}
Let $\mathcal{M}:\mathfrak{H} \leadsto \mathfrak{H}$ a quantum channel. Relative to some error basis $\mathcal{E}$, its circuit tensor $\circuitoperator{\mathcal{M}}$ and process matrix $\chi$ satisfy $\circuitoperator{\mathcal{M}} = \dim(\mathfrak{H}) \Psi(\chi)$.
\end{theorem}
\begin{proof}
    Let $\chi$ be the process matrix of $\mathcal{M}$, and via diagonalization write $\chi^E_{E'} = \sum_j \lambda_{j,E'}\overline{\lambda_{j,E}}$. Then as above define Kraus operators $A_j = \sum_E \lambda_{j,E} E^\dagger$. Finally, compute
    \begin{align*}
        \Psi(\chi) &= \sum_{E,E'\in\mathcal{E}} \chi\errorbasis{E}{E'} \Psi(e\errorbasis{E}{E'})\\
        &= \frac{1}{\dim(\mathfrak{H})^2} \sum_{E,E',F,F'} \chi^E_{E'} \Tr(F^\dagger E F' (E')^\dagger) e\errorbasis{F}{F'}\\
        &= \frac{1}{\dim(\mathfrak{H})^2} \sum_{j,E,E',F,F'} \lambda_{j,E'}\overline{\lambda_{j,E}} \Tr(F^\dagger E F' (E')^\dagger) e\errorbasis{F}{F'}\\
        &= \frac{1}{\dim(\mathfrak{H})^2} \sum_{j,F,F'} \Tr(F^\dagger A_j^\dagger F' A_j) e\errorbasis{F}{F'}\\
        &= \frac{1}{\dim(\mathfrak{H})} \circuitoperator{\mathcal{M}}.
    \end{align*}
\end{proof}

\end{document}